\documentclass[onefignum,onetabnum]{siamart190516}
\usepackage{amsmath,rotating}
\usepackage{amssymb,verbatim}
\usepackage{amsfonts,todonotes}
\usepackage{graphicx,todonotes}
\usepackage[font={small}]{caption}
\usepackage{microtype,boxedminipage}
\usepackage[rightcaption]{sidecap}    
\usepackage{enumitem}
\usepackage{hyperref,microtype}
\usepackage{booktabs}

\usetikzlibrary{shapes}

\newcommand\blfootnote[1]{%
  \begingroup
  \renewcommand\thefootnote{}\footnote{#1}%
  \addtocounter{footnote}{-1}%
  \endgroup
}

\newcommand{\cD}{{\cal D}}

\newcommand{\cF}{{\cal F}}

\newcommand{\cH}{{\cal H}}
\newcommand{\cI}{{\cal I}}
\newcommand{\cJ}{{\cal J}}

\newcommand{\cL}{{\cal L}}

\newcommand{\cP}{{\cal P}}
\newcommand{\cQ}{{\cal Q}}

\newcommand{\cY}{{\cal Y}}

\newcommand{\CCC}{\mathcal{C}}

\newcommand{\bigoh}{\ensuremath{{\mathcal O}}}

\newcommand{\cvc}{tcw\textsc{-CVC}}
\newcommand{\CVC}{tcw\textsc{-CVC}}
\newcommand{\imb}{tcw\textsc{-IMB}}
\newcommand{\cds}{tcw\textsc{-CDS}}

\newcommand{\fpt}{\mathsf{FPT}}
\newcommand{\xp}{\mathsf{XP}}
\newcommand{\np}{\mathsf{NP}}
\newcommand{\wx}[1]{\mathsf{W}[#1]}

\newcommand{\mc}{\mathcal}
\newcommand{\SB}{\{\,} \newcommand{\SM}{\;{|}\;} \newcommand{\SE}{\,\}}
\newtheorem{observation}{Observation}

\newcommand{\W}[1][xxxx]{\text{\normalfont W[#1]}}
\newcommand{\XP}{\text{\normalfont XP}}

\newcommand{\pw}{{\mathbf{pw}}}
\newcommand{\td}{{\mathbf{td}}}
\newcommand{\tcw}{{\mathbf{tcw}}}
\newcommand{\tw}{{\mathbf{tw}}}
\newcommand{\cut}{{\mathsf{cut}}}
\newcommand{\adh}{{\mathsf{adh}}}
\newcommand{\tor}{{\mathsf{tor}}}
\newcommand{\cvcv}{\mathsf{cvc}}
\newcommand{\imbv}{\mathsf{imb}}
\newcommand{\cost}{\mathsf{cost}}
\newcommand{\acost}{\mathsf{Acost}}

\newcommand{\PP}{P}
\newcommand{\QQ}{Q}
\newcommand{\Nat}{\mathbb{N}}

\title{Algorithmic Applications of Tree-Cut Width}

\begin{document}

\newtheorem{fact}{Fact}
\newtheorem{notation}{Notation}
\newtheorem{claimm}{Claim}

\author{Robert Ganian\thanks{Algorithms and Complexity Group, TU Wien, Vienna, Austria.}
\and Eun Jung Kim\thanks{CNRS, Universit\'{e} Paris-Dauphine, Paris, France.}
\and Stefan Szeider\footnotemark[2]}

%


\maketitle

\begin{abstract}
  The recently introduced graph parameter \emph{tree-cut width} plays
  a similar role with respect to immersions as the graph parameter
  \emph{treewidth} plays with respect to minors. In this paper, we
  provide the first algorithmic applications of tree-cut width to hard
  combinatorial problems. Tree-cut width is known to be lower-bounded
  by a function of treewidth, but it can be much larger and hence has
  the potential to facilitate the efficient solution of problems that
  are not known to be fixed-parameter tractable (FPT) when
  parameterized by treewidth. We introduce the notion of nice tree-cut
  decompositions and provide FPT algorithms for the showcase problems
  \textsc{Capacitated Vertex Cover}, \textsc{Capacitated Dominating
    Set}, and \textsc{Imbalance} parameterized by the tree-cut width of
  an input graph. On the other hand, we show that \textsc{List
    Coloring}, \textsc{Precoloring Extension}, and \textsc{Boolean CSP}
  (the latter parameterized by the tree-cut width of the incidence
  graph) are W[1]-hard and hence unlikely to be fixed-parameter
  tractable when parameterized by tree-cut width.
  \blfootnote{\textnormal
    An extended
  abstract of this paper appeared in the proceedings of MFCS 2015, the
  40th International Symposium on Mathematical Foundations of Computer Science.}  
\end{abstract}



\section{Introduction}
In their seminal work on graph minors,
Robertson and Seymour have shown that all finite graphs are not only
well-quasi ordered by the \emph{minor} relation, but also by the
\emph{immersion} relation\footnote{A graph $H$ is an immersion of a
  graph $G$ if $H$ can be obtained from $G$ by applications of vertex
  deletion, edge deletion, and edge lifting, i.e.,
  replacing two incident edges by a single edge which joins the two
  vertices not shared by the two edges.}, the Graph Immersion Theorem
\cite{RobertsonSeymour95}. This verified a conjecture by
Nash-Williams~\cite{Nash64,Nash65}.  As a consequence of this theorem, each
graph class that is closed under taking immersions can be
characterized by a finite set of forbidden immersions, in analogy to a
graph class closed under taking minors being characterized by a finite
set of forbidden minors.

In a recent paper \cite{Wollan15}, Wollan introduced the graph
parameter \emph{tree-cut width}, which plays a similar role with
respect to immersions as the graph parameter \emph{treewidth} plays
with respect to minors. Wollan obtained an analog to the
Excluded Grid Theorem for these notions: if a graph has bounded
tree-cut width, then it does not admit an immersion of the $r$-wall
for arbitrarily large~$r$~\cite[Theorem~15]{Wollan15}. Marx
and Wollan~\cite{MarxWollan14} proved that for all $n$-vertex graphs
$H$ with maximum degree $k$ and all $k$-edge-connected graphs $G$,
either $H$ is an immersion of $G$, or $G$ has tree-cut width
bounded by a function of $k$ and~$n$.

In this paper, we provide the first algorithmic applications of tree-cut
width to hard combinatorial problems. Tree-cut width is known to be
lower-bounded by a function of treewidth, but it can be much larger
than treewidth if the maximum degree is unbounded (see Subsection~\ref{sub:Rel} for an comparison of
tree-cut width to other parameters). Hence tree-cut width has the
potential to facilitate the efficient solution of problems which are
not known or not believed to be fixed-parameter tractable (FPT) when parameterized
by treewidth. For other problems it might allow the strengthening of parameterized hardness results.

We provide results for both possible outcomes: in
Section~\ref{sec:algorithms} we provide FPT algorithms for the
showcase problems \textsc{Capacitated Vertex Cover},
\textsc{Capacitated Dominating Set} and \textsc{Imbalance}
parameterized by the tree-cut width of an input graph $G$, while in
Section~\ref{sec:hard} we show that \textsc{List Coloring},
\textsc{Precoloring Extension} and \textsc{Boolean CSP} parameterized
by tree-cut width (or, for the third problem, by the tree-cut width of
the incidence graph) are not likely to be
FPT. Table~\ref{tab:overview} provides an overview of our results. The
table shows how tree-cut width provides an intermediate measurement
that allows us to push the frontier for fixed-parameter tractability
in some cases, and to strengthen $\W[1]$-hardness results in some
other cases.

\newcommand{\super}[1]{${}^\text{#1}$}

\begin{table}[htb]
  \centering 
\resizebox{\columnwidth}{!}{
  \begin{tabular}{@{}ll@{\hspace{2mm}}l@{\hspace{2mm}}l@{}}
 \toprule
     & \multicolumn{3}{c}{\emph{Parameter}} \\ \cmidrule(){2-4}

     \emph{Problem}                    & \emph{tw} & \emph{tree-cut width} &
                                      \emph{max-degree and tw}   \\
    \textsc{Capacitated Vertex Cover} & $\W[1]$-hard\super{\cite{DomLokstanovSaurabhVillanger08}}
 & FPT\super{(Thm \ref{thm:CVCalg})} & FPT \\ 
    \textsc{Capacitated Dominating Set} &
    $\W[1]$-hard\super{\cite{DomLokstanovSaurabhVillanger08}} &
    FPT\super{(Thm \ref{thm:CDSPalg})} & FPT \\ 
    \textsc{Imbalance} & Open\super{\cite{LokshtanovMisraSaurabh13}} & FPT\super{(Thm~\ref{thm:IMBalg})} & FPT\super{\cite{LokshtanovMisraSaurabh13}} \\ 
    \textsc{List Coloring} & $\W[1]$-hard\super{\cite{FellowsEtAl11}}
    & $\W[1]$-hard\super{(Thm \ref{thm:listhard})} & FPT\super{(Obs~\ref{obs:coldeg})} \\ 
    \textsc{Precoloring Extension} & $\W[1]$-hard\super{\cite{FellowsEtAl11}} & $\W[1]$-hard\super{(Thm \ref{thm:listhard})} & FPT\super{(Obs~\ref{obs:coldeg})} \\ 
    \textsc{Boolean CSP} & $\W[1]$-hard\super{\cite{SamerSzeider10a}} & $\W[1]$-hard\super{(Thm \ref{thm:csphard})} & FPT\super{\cite{SamerSzeider10a}} \\ 
\bottomrule
  \end{tabular}
  }
  \caption{Overview of results (\emph{tw} stands for treewidth).  \vspace{-0.6cm}
   }
  \label{tab:overview}
\end{table}

Our FPT algorithms assume that a suitable decomposition, specifically
a so-called \emph{tree-cut decomposition}, is given as part of the
input. Since the class of graphs of tree-cut width at most $k$ is closed
under taking immersions~\cite[Lemma~10]{Wollan15}, the Graph Immersion
Theorem together with the fact that immersions testing is fixed-parameter tractable~\cite{GroheKawarabayashiMarxWollan11} gives rise
to a non-uniform FPT algorithm for testing whether a graph has
tree-cut width at most~$k$. Kim et
al.~\cite{KimOumPaulSauThilikos14} provide a uniform FPT algorithm
which constructs a tree-cut decomposition whose width is at most twice
the optimal one.  Effectively, this result allows us to remove the
condition that a tree-cut decomposition is supplied as part of the
input from our uniform FPT algorithms.

We briefly outline the methodology used to obtain our algorithmic
results. As a first step, in Section~\ref{sec:nice} we develop the
notion of \emph{nice tree-cut decompositions}\footnote{We call them ``nice''
as they serve a similar purpose as the nice tree decompositions~\cite{Kloks94},
although the definitions are completely unrelated.} and show that every
tree-cut decomposition can be transformed into a nice one in
polynomial time. These nice tree-cut decompositions are of independent
interest, since they provide a means of simplifying the complex
structure of tree-cut decompositions. In Section~\ref{sec:algorithms}
we introduce a general three-stage framework for the design of
FPT algorithms on nice tree-cut decompositions and apply it to our
problems. The crucial part of this framework is the computation of the
``joins''. We show that the children of any node in a nice tree-cut
decomposition can be partitioned into (i)~a bounded number of children
with complex connections to the remainder of the graph, and (ii)~a
potentially large set of children with only simple connections to the
remainder of the graph. We then process these by a combination of
branching techniques applied to~(i) and integer linear programming
applied to~(ii). The specifics of these procedures differ from
problem to problem.

\smallskip


\section{Preliminaries}
\label{sec:prelim}

\subsection{Basic Notation} 
We use standard terminology for graph theory, see for instance~\cite{Diestel00}. All graphs except for those used to compute the torso-size in Subsection~\ref{sub:tcw} are simple; the multigraphs used in Subsection~\ref{sub:tcw} have loops, and each loop increases the degree of the vertex by $2$.

Given a graph $G$, we let $V(G)$ denote its vertex set and $E(G)$ its edge set. The (open) neighborhood of a vertex $x \in V(G)$ is the set $\{y\in V(G):xy\in E(G)\}$ and is denoted by $N_G(x)$. The closed neighborhood $N_G[x]$ of $x$ is defined as $N_G(x)\cup \{x\}$. For a vertex subset $X$, the (open) neighborhood of $X$ is defined as $\bigcup_{x\in X} N_G(x) \setminus X$ and denoted by $N_G(X)$. The set $N_G[X]$ refers to the closed neighborhood of $X$ defined as $N_G(X)\cup X$. We refer to the set $N_G(V(G)\setminus X)$ as $\partial_G(X)$; this is the set of vertices in $X$ which have a neighbor in $V(G)\setminus X$. The degree of a vertex $v$ in $G$ is denoted $deg_G(v)$, and a vertex of degree $1$ is called a \emph{pendant vertex}. 
When the graph we refer to is clear, we drop the lower index $G$ from the notation. We use $[i]$ to denote the set $\{0,1,\dots,i\}$.


\subsection{Parameterized Complexity}

A \emph{parameterized problem} $\PP$ is a subset of $\Sigma^* \times \Nat$ for some finite alphabet $\Sigma$. Let $L\subseteq \Sigma^*$ be a classical decision problem for a finite alphabet, and let $p$ be a non-negative integer-valued function defined on $\Sigma^*$. Then $L$ \emph{parameterized by} $p$ denotes the parameterized problem $\SB(x,p(x)) \SM x\in L \SE$ where $x\in \Sigma^*$.  For a problem instance $(x,k) \in \Sigma^* \times \Nat$ we call $x$ the main part and $k$ the parameter.  
A parameterized problem $\PP$ is \emph{fixed-parameter   tractable} (FPT in short) if a given instance $(x, k)$ can be solved in time  $f(k) \cdot |x|^{\bigoh(1)}$ where $f$ is an arbitrary computable function of $k$. 


Parameterized complexity classes are defined with respect to {\em fpt-reducibility}. A parameterized problem $\PP$ is {\em fpt-reducible} to $\QQ$ if in time $f(k)\cdot |x|^{\bigoh(1)}$, one can transform an instance $(x,k)$ of $\PP$ into an instance $(x',k')$ of $\QQ$ such that $(x,k)\in \PP$ if and only if $(x',k')\in \QQ$, and $k'\leq g(k)$, where $f$ and $g$ are computable functions depending only on $k$.  We also use the related notion of \emph{FPT Turing reductions}~\cite{FlumGrohe06,DowneyFellows13,CyganFKLMPPS15}. 
For two parameterized problems $P$ and $Q$, an FPT Turing reduction from $\PP$ to $\QQ$ is an FPT algorithm which decides whether $(x,k)$ is in $\PP$, 
provided that there is an oracle access to $\QQ$ 
 for all instances $(x',k')$ of $\QQ$ with $k'\leq g(k)$ for some computable function $g$.
Informally speaking, a problem $\PP$ is admits an FPT  Turing reduction to $\QQ$ if one can solve an instance of $\PP$ by an FPT algorithm 
which has access to an oracle which can solve $\QQ$ but which may only be used for instances of $Q$ whose parameter is upper-bounded by a function of the parameter of $P$. 
 Owing to the definition, if $\PP$ fpt-reduces or FPT Turing-reduces to $\QQ$
and $\QQ$ is fixed-parameter tractable, then $P$ is fixed-parameter
tractable as well.

Central to parameterized complexity is the following hierarchy of complexity classes, defined by the closure of canonical problems under fpt-reductions:
\[\fpt \subseteq \wx{1} \subseteq \wx{2} \subseteq \cdots \subseteq \xp.\] All inclusions are believed to be strict. In particular, $\fpt\neq \wx{1}$ under the Exponential Time Hypothesis~\cite{ImpagliazzoPaturiZane01}. 

The class $\wx{1}$ is the analog of $\np$ in parameterized complexity. A major goal in parameterized complexity is to distinguish between parameterized problems which are in $\fpt$ and those which are $\wx{1}$-hard, i.e., those to which every problem in $\wx{1}$ is fpt-reducible. There are many problems shown to be complete for $\wx{1}$, or equivalently $\wx{1}$-complete, including the {\sc Multi-Colored Clique (MCC)} problem~\cite{DowneyFellows13}.

\subsection{Integer Linear Programming}
Our algorithms use an Integer Linear Programming (ILP) subroutine. ILP is a well-known framework for formulating problems and a powerful tool for the development of fpt-algorithms for optimization problems. 

\begin{definition}[$p$-Variable Integer Linear Programming Optimization] Let $A\in \mathbb{Z}^{q\times p}, b\in \mathbb{Z}^{q\times 1}$ and $c\in \mathbb{Z}^{1\times p}$. The task is to find a vector $x\in \mathbb{Z}^{p\times 1}$ which minimizes the objective function $c\times x$ and satisfies all $q$ inequalities given by $A$ and $b$, specifically satisfies $A\cdot x\geq b$. The number of variables $p$ is the parameter.
\end{definition}

Lenstra \cite{Lenstra83} showed that \textsc{$p$-ILP}, together with its optimization variant \textsc{$p$-OPT-ILP} (defined above), are in FPT. His running time was subsequently improved by Kannan \cite{Kannan87} and Frank and Tardos \cite{FrankTardos87} (see also \cite{FellowsLokshtanovMisraRS08}).

\begin{theorem}[\cite{Lenstra83,Kannan87,FrankTardos87,FellowsLokshtanovMisraRS08}]
\label{thm:pilp}
\textsc{$p$-OPT-ILP} can be solved using $\bigoh(p^{2.5p+o(p)}\cdot L)$ arithmetic operations in space polynomial in $L$, $L$ being the number of bits in the input.
\end{theorem}

\subsection{Tree-Cut Width}
\label{sub:tcw}
The notion of tree-cut decompositions was first proposed by Wollan~\cite{Wollan15}, see also~\cite{MarxWollan14}.
A family of subsets $X_1, \ldots, X_{k}$ of $X$ is a {\em near-partition} of $X$ if they are pairwise disjoint and $\bigcup_{i=1}^{k} X_i=X$, allowing the possibility of $X_i=\emptyset$.  

\begin{definition}
A {\em tree-cut decomposition} of $G$ is a pair $(T,\mc{X})$ which consists of a tree $T$ and a near-partition $\mc{X}=\{X_t\subseteq V(G): t\in V(T)\}$ of $V(G)$. A set in the family $\mc{X}$ is called a {\em bag} of the tree-cut decomposition. 
\end{definition}

For any edge $e=(u,v)$ of $T$, let $\Upsilon_u$ and $\Upsilon_v$ be the two connected components in $T-e$ which contain $u$ and $v$ respectively. Note that $(\bigcup_{t\in \Upsilon_u} X_t, \bigcup_{t\in \Upsilon_v} X_t)$ is a partition of $V(G)$, and we use $\cut(e)$ to denote the set of edges with one endpoint in each partition. A tree-cut decomposition is \emph{rooted} if one of its nodes is called the root, denoted by $r$. For any node $t\in V(T) \setminus \{r\}$, let $e(t)$ be the unique edge incident to $t$ on the path between $r$ and $t$. We define the {\em adhesion} of $t$ ($\adh_T(t)$ or $\adh(t)$ in brief) as $|\cut(e(t))|$; if $t$ is the root, we set $\adh_T(t)=0$.

The {\em torso} of a tree-cut decomposition $(T,\mc{X})$ at a node $t$, written as $H_t$, is the graph obtained from $G$ as follows. If $T$ consists of a single node $t$, then the torso of $(T,\mc{X})$ at $t$ is $G$. Otherwise let $T_1, \ldots , T_{\ell}$ be the connected components of $T-t$. For each $i=1,\ldots , \ell$, the vertex set $Z_i$ of $V(G)$ is defined as the set $\bigcup_{b\in V(T_i)}X_b$. The torso $H_t$ at $t$ is obtained from $G$ by {\em consolidating} each vertex set $Z_i$ into a single vertex~$z_i$. Here, the operation of consolidating a vertex set $Z$ into $z$ is to substitute $Z$ by $z$ in $G$, and for each edge $e$ between $Z$ and $v\in V(G)\setminus Z$, adding an edge $zv$ in the new graph. We note that this may create parallel edges.

The operation of {\em suppressing} a vertex $v$ of degree at most $2$ consists of deleting~$v$, and when the degree is two, adding an edge between the neighbors of $v$. Given a graph $G$ and  $X\subseteq V(G)$, let the {\em 3-center} of $(G,X)$ be the unique graph obtained from $G$ by exhaustively suppressing vertices in $V(G) \setminus X$ of degree at most two. Finally, for a node $t$ of $T$, we denote by $\tilde{H}_t$ the 3-center of $(H_t,X_t)$, where $H_t$ is the torso of $(T,\mc{X})$ at $t$. 
Let the \emph{torso-size} $\tor(t)$ denote $|\tilde{H}_t|$. 

\begin{definition}
The width of a tree-cut decomposition $(T,\mc{X})$ of $G$ is $\max_{t\in V(T)}\{ \adh(t), \tor(t) \}$. The tree-cut width of $G$, or $\tcw(G)$ in short, is the minimum width of $(T,\mc{X})$ over all tree-cut decompositions $(T,\mc{X})$ of $G$.
\end{definition}

We conclude this subsection with some notation related to rooted tree-cut decompositions. 
For $t\in V(T)\setminus \{r\}$, we let $p_T(t)$ (or $p(t)$ in brief) denote the parent of $t$ in $T$. For two distinct nodes $t,t'\in V(T)$, we say that $t$ and $t'$ are {\em siblings} if $p(t)=p(t')$. Given a tree node $t$, let $T_t$ be the subtree of $T$ rooted at $t$. Let $Y_t=\bigcup_{b\in V(T_t)} X_b$, and let $G_t$  denote the induced subgraph $G[Y_t]$. The \emph{depth} of a node $t$ in $T$ is the distance of $t$ from the root $r$.
The vertices of $\partial(Y_t)$ are called the {\em borders} at node $t$.  A node $t\neq r$ in a rooted tree-cut decomposition is \emph{thin} if $\adh(t)\leq 2$ and \emph{bold} otherwise. 

\begin{figure}[ht]
\begin{center}
\begin{tikzpicture}
\tikzstyle{every node}=[draw, shape=circle, minimum size=3pt,inner sep=2pt, fill=white]

\draw (0,0) node[label=left:$a$] (a){}; 
\draw (1,0) node [label=270:$d$] (d){};
\draw (0,1) node[label=left:$b$] (b){}; 
\draw (1,1) node [label=right:$c$] (c){};
\draw (2,0) node [label=right:$e$](e){}; \draw (2,0.5) node [label=right:$f$](f){};
\draw (2,1) node [label=right:$g$](g){};
\draw (b)--(d)--(a)--(b)--(c)--(d)--(e)--(f); \draw (f)--(d)--(g);
\end{tikzpicture}
\quad\quad
\begin{tikzpicture}[scale=0.9]
\tikzstyle{cloud} = [draw, ellipse, minimum height=0em, minimum width=3em]

\node (d) at (0,2)[cloud,label=left:{$(2,0)$}] {d};
\node (a) at (-1,1)[cloud,label=left:{$(3,3)$}] {a};
\node (bc) at (-1,0)[cloud,label=left:{$(3,3)$}] {bc};
\node (e) at (0.5,1)[cloud,label=270:{$(1,2)$}] {e};
\node (f) at (2,1)[cloud,label=270:{$(1,2)$}] {f};
\node (g) at (3.5,1)[cloud,label=270:{$(1,1)$}] {g};

\draw[ultra thick] (bc)--(a)--(d);
\draw (d)--(e); \draw (f)--(d)--(g);

\tikzstyle{empty}=[draw, shape=circle, minimum size=3pt, inner sep=0pt, fill=white, color=white]


\end{tikzpicture}
\end{center}
\caption{A graph $G$ and a width-$3$ tree-cut decomposition of $G$, including the torso-size (left value) and adhesion (right value) of each node.}
\vspace{-0.5cm}
\end{figure}
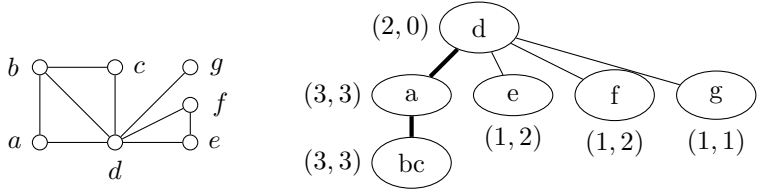


\subsection{Relations to Other Width Parameters}
\label{sub:Rel}
Here we review the relations between the tree-cut width and other
width parameters, specifically \emph{treewidth} ($\mathbf{tw}$),
\emph{pathwidth} ($\mathbf{pw}$)~\cite{Diestel00}, 
and \emph{treedepth} ($\mathbf{td}$)
\cite{NesetrilMendez06}. We also compare to the maximum over treewidth
and maximum degree, which we refer to as \emph{degree treewidth}
($\mathbf{degtw}$).

\begin{proposition}[\cite{Wollan15,MarxWollan14}]\label{prop:comparison}
There exists a function $h$ such that $\tw(G)\leq h(\tcw(G))$ and $\tcw(G)\leq 4\mathbf{degtw}(G)^2$ for any graph $G$.
\end{proposition}

Below, we provide a new explicit bound on the relationship between treewidth and tree-cut width, and show that it is incomparable with pathwidth and treedepth. Since the proof of the following Proposition~\ref{prop:twtcw} relies on Theorem~\ref{thm:decompsize} and notation introduced later, we postpone its proof to Section~\ref{sec:algorithms}.

\begin{proposition}
\label{prop:twtcw}
For any graph $G$ it holds that $\tw(G)\leq 2\tcw(G)^2+3\tcw(G)$.
\end{proposition}

\begin{proposition}
\label{prop:tcwtdpw}
There exists a graph class $\cH_1$ of bounded pathwidth and treedepth but unbounded tree-cut width, and there exists a graph class $\cH_2$ of bounded tree-cut width but unbounded pathwidth and treedepth.
\end{proposition}

\begin{proof}
Let $\cH_1$ be the class of full ternary trees. It is easy to see that, for each ternary tree $T_n$ of depth $n$, $\pw(T_n)\in \Theta(n)$ 
and $\td(T_n)\in \Theta(n)$~\cite{TAKAHASHI1994}. On the other hand, the tree-cut width of each $T_n$ is bounded by a constant by Proposition~\ref{prop:comparison}. 
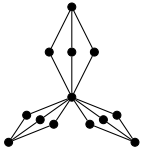
\begin{figure}[th]
\begin{flushleft}
\begin{center}
\begin{tikzpicture}[scale=1.2]
\tikzstyle{every node}=[draw, shape=circle, minimum size=3pt,inner sep=1pt, fill=black]

\draw (0,0) node (x){};
\draw (-0.7,-0.5) node (a){};
\draw (0.7,-0.5) node (b){};
\draw (0,1) node (c){};

\draw (-0.35,-0.25) node (a1){};
\draw (-0.5,-0.2) node (a2){};
\draw (-0.2,-0.3) node (a3){};

\draw (0.35,-0.25) node (b1){};
\draw (0.5,-0.2) node (b2){};
\draw (0.2,-0.3) node (b3){};

\draw (0,0.5) node (c1){};
\draw (-0.25,0.5) node (c2){};
\draw (0.25,0.5) node (c3){};

\draw (x)--(c1)--(c); \draw (x)--(c2)--(c); \draw (x)--(c3)--(c);
\draw (x)--(a1)--(a); \draw (x)--(a2)--(a); \draw (x)--(a3)--(a);
\draw (x)--(b1)--(b); \draw (x)--(b2)--(b); \draw (x)--(b3)--(b);

\end{tikzpicture}
\end{center}
\end{flushleft}
\caption{The graph $S_3$.}
\end{figure}
Let $\cH_2$ be the class of graphs $S_n$ obtained from a star with $n$ leaves $z_1,\ldots , z_n$ by replacing each edge with $n$ subdivided edges. It is easy to see that $\pw(S_n)=2=\td(S_n)-1$. We verify that $\tcw(S_n)\geq n$: suppose $\tcw(S_n)\leq n-1$. Then any two leaves $z_i$ and $z_j$, $i\neq j$, must be contained in the same bag in any tree-cut decomposition of width at most $n-1$ as they are connected by $n$ edge-disjoint paths. This means there exists a bag $t$ containing all $z_i$'s in any such tree-cut decomposition, which however implies that $\tor(t)\geq n$.
\end{proof}

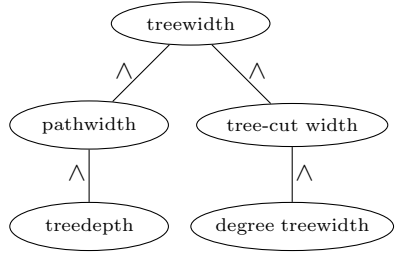
\begin{figure}[h]
\begin{center}
\vspace{-0.4cm}
\begin{tikzpicture}[scale=0.9]
\tikzstyle{cloud} = [draw, ellipse, minimum height=0em, minimum width=6em]

\node (tcw) at (1.5,3.5)[cloud] {\scriptsize \!\!tree-cut width\!\!};
\node (pw) at (-1.5,3.5)[cloud] {\scriptsize pathwidth};
\node (td) at (-1.5,2)[cloud] {\scriptsize treedepth};
\node (tw) at (0,5)[cloud] {\scriptsize treewidth};
\node (dtw) at (1.5,2)[cloud] {\scriptsize \!\!\!\!degree treewidth\!\!\!\!};

\draw (dtw) -- (tcw) -- (tw);
\draw (td) -- (pw) -- (tw);

\tikzstyle{every node}=[]
\tikzstyle{empty}=[draw, shape=circle, minimum size=3pt,inner sep=0pt, fill=white, color=white]

\node (vk) at (1.3,4.3)[empty, label=left:\begin{sideways}$>$\end{sideways}] {};
\node (lr) at (-1.3,4.3)[empty, label=right:\begin{sideways}$>$\end{sideways}] {};
\node (aa) at (2,2.8)[empty, label=left:\begin{sideways}$>$\end{sideways}] {};
\node (bb) at (-2,2.8)[empty, label=right:\begin{sideways}$>$\end{sideways}] {};

\end{tikzpicture}
\end{center}
\caption[center]{Relationships between selected graph parameters
  ($A\!\! > \!\!B$ means that every graph class of bounded $A$ is
  also of bounded $B$, but there are graph classes of bounded $B$
  which are not of bounded $A$).}
  \vspace{-0.5cm}
\end{figure}

\section{Nice Tree-Cut Decompositions}
\label{sec:nice}

In this section we introduce the notion of a {\em nice tree-cut decomposition} and present an algorithm to transform any tree-cut decomposition into a nice one without increasing the width. 
A tree-cut decomposition $(T,\mc{X})$ rooted at $r$ is \emph{nice} if it satisfies the following condition for every thin node $t\in V(T)$:
$N(Y_t)\cap \bigcup_{b\text{ is a sibling of }t}Y_b=\emptyset.$

\smallskip
The intuition behind nice tree-cut decompositions is that we restrict the neighborhood of thin nodes in a way which facilitates dynamic programming. 

\begin{lemma}
\label{lem:nicetcdec}
There exists a cubic-time algorithm which transforms any rooted tree-cut decomposition $(T,\mc{X})$ of $G$ into a nice tree-cut decomposition of the same graph, without increasing its width or number of nodes.
\end{lemma}

The proof of Lemma~\ref{lem:nicetcdec} is based on the following considerations.
Let $(T,\mc{X})$ be a rooted tree-cut decomposition of $G$ whose width is at most $w$. We say that a node $t$, $t \neq r$, is {\em bad} if it violates the above condition, i.e., $\adh(t)\leq 2$ and there is a sibling $b$ of $t$ such that $N(Y_t)\cap Y_b\neq \emptyset$. For a bad node $t$, we say that $b$ is a {\em bad neighbor} of $t$ if $N(Y_t) \cap X_b\neq \emptyset$ and $b$ is either a sibling of $t$ or a descendant of a sibling of $t$. 
In order to construct a tree-cut decomposition with the claimed property, we introduce a rerouting operation in which we pick a bad node 
and relocate it, that is, change the parent of the selected bad node. 

\begin{quote}
\textsc{Rerouting($t$)}: let $t$ be a bad node and let $b$ be a bad neighbor of $t$ of maximum depth (resolve ties arbitrarily). Then remove the tree edge $e(t)$ from $T$ and add a new tree edge $\{b,t\}$.
\end{quote}
%

\begin{quote}
{\sc Top-down Rerouting}: as long as $(T,\mc{X})$ is not a nice tree-cut decomposition, pick any bad node $t$ of minimum depth. Perform {\sc Rerouting($t$)}.
\end{quote}
The selection criteria of a new parent in \textsc{Rerouting($t$)} ensures 
that no new bad node is created at the depth of $p(t)$ or at a higher depth. This, together with the priority given to a bad node of minimum depth in {\sc Top-down Rerouting},  
implies that the tree gradually becomes free from bad nodes in a top-down manner.

\begin{lemma}\label{lem:rerout}
The procedure {\sc Rerouting($t$)} does not increase the width of the tree-cut decomposition.
\end{lemma}

\begin{proof}
Let $b$ be the bad neighbor of $t$ chosen by \textsc{Rerouting($t$)}, and let $(T',\mc{X})$ be the tree-cut decomposition obtained from the tree-cut decomposition $(T,\mc{X})$ by \textsc{Rerouting($t$)}. We will show that for each $z\in V(T)$, it holds that 
\begin{enumerate}
\item[(1)] $\adh_T(z)\geq \adh_{T'}(z)$, and 
\item[(2)] $\tor_T(z)\geq \tor_{T'}(z)$, 
\end{enumerate}
from which the lemma follows.

Let us first consider Claim (1). Let $P$ be the set of edges on the path in $T$ between $b$ and $p(t)$. Slightly abusing the notation, 
we also denote the path in $T'$ connecting $b$  and $p(t)$ by $P$. Then for any edge $p$ of $T'$ 
which is not on $P$, it holds that $\cut_T(p)=\cut_{T'}(p)$, and hence Claim (1) holds for all nodes $z$ which are not on $P$ and also for $z=p(t)$. 
As for the remaining nodes $z$ on $P$, note that $z\neq t$ and it holds that every edge between $X_b$ and $X_t$ lies in $\cut_T(e(z))\setminus \cut_{T'}(e(z))$. 
Furthermore, by thinness of $t$ and the fact that $X_t$ has a neighbor in $X_b$, there may exist at most one edge $e'$ such that $e' \in \cut_{T'}(e(z))\setminus \cut_T(e(z))$, and hence either $\adh_T(z)= \adh_{T'}(z)$ or $\adh_T(z)= \adh_{T'}(z)+1$. Finally, note that $\adh_T(t)= \adh_{T'}(t)$. Thus Claim (1) holds.

Now we consider Claim (2). From Claim (1) it follows that Claim (2) may only be violated if $N_T(z)\subset N_{T'}(z)$, which is only the case for $z=b$. However, it is easy to verify that $\tor_T(b)\geq \tor_{T'}(b)$ since $\{t\}=N_{T'}(b)\setminus N_{T}(b)$ and $t$ is thin.
\end{proof}

\begin{lemma}
\label{lem:topdown}
{\sc Top-down Rerouting} terminates after performing {\sc Rerouting} at most $2|T|$ times, where $(T,\mc{X})$ is the initial tree-cut decomposition.
\end{lemma}

\begin{proof}
The key observation is that 
once a bad node $t$ is rerouted to be a child of its bad neighbor $b$, $b$ will remain as an ancestor of $t$ in the  tree-cut decompositions obtained 
by the subsequent {\sc Top-down Rerouting}
and thus $b$ will not become a bad neighbor of $t$ again. Indeed, the only way that $b$ ceases to be a parent of $t$ is when $t$ is rerouted again as a bad node. 
Yet, the bad neighbor of $t$ is a descendant of $b$ and thus after the (second) rerouting $t$ remains as a descendant of $b$. 
Furthermore, this implies that once a bad node $t$ stops being a bad node, $t$ will never turn into a bad node in the subsequent {\sc Top-down Rerouting}.

To see that {\sc Top-down Rerouting} terminates in at most $2|T|$ steps, consider maintaining the following auxiliary binary relation over the nodes of $T$ during the procedure. 
Let $(t,b)$ be a \emph{bad pair} if $t$ is a bad node of $T$, $b$ is a bad neighbor of $t$, and there is an edge of $G$ between $X_t$ and $X_b$. Let $B$ be the binary relation containing all bad pairs.
We update the binary relation $B$ as we perform {\sc Top-down Rerouting} so that $B$ always consists of the bad pairs in the current tree-cut decomposition. 
When a bad neighbor $b$ of $t$ becomes the parent of $t$ by rerouting, we eliminate $(t,b)$ from $B$. 
Therefore, the number of pairs ever eliminated from $B$ throughout  {\sc Top-down Rerouting} provides an upper bound 
on the number of steps. Also the number of pairs eliminated from $B$ equals the number of pairs which were ever inserted into $B.$ 
The latter is bounded by $2|T|$ because the number of pairs which ever enter $B$ with the first entry $t$ is at most two (due to $\adh(t)\leq 2$)
and any pair which was removed from $B$ will never enter it again. It follows that {\sc Top-down Rerouting} terminates in at most $2|T|$ steps.
\end{proof}

\begin{proof}[Proof of Lemma~\ref{lem:nicetcdec}]
By definition, the output of \textsc{Top-down Rerouting} is a nice tree-cut decomposition. The lemma then follows from Lemma~\ref{lem:topdown} and Lemma~\ref{lem:rerout}.
\end{proof}


The following Theorem~\ref{thm:decompsize} builds upon Lemma~\ref{lem:nicetcdec} by additionally giving a bound on the size of the decomposition.

\begin{theorem}
\label{thm:decompsize}
If $\tcw(G)= k$, then there exists a nice tree-cut decomposition $(T, \mc{X})$ of $G$ of width $k$ with at most $2|V(G)|$ nodes. Furthermore, $(T,\mc{X})$ can be computed from any width-$k$ tree-cut decomposition of $G$ in quadratic time.
\end{theorem}

\begin{proof}
For brevity, we say that a node $t$ is empty if $X_t=\emptyset$.
Consider a rooted width-$k$ tree-cut decomposition $(T,\mc{X})$ of $G$. We can assume that each leaf $t$ of $T$ is non-empty, since otherwise $t$ may be deleted from $T$ without changing the width. 

Consider an empty non-leaf node $t$ which has a single child $t_1$. Let $(T',\mc{X})$ be the tree-cut decomposition obtained by suppressing the node $t$. We claim that the width of $(T',\mc{X})$ is at most $k$. Indeed, $\adh_{T'}(t_1)=\adh_T(t_1)=\adh_T(t)$ and $\adh_{T'}(u)=\adh_T(u)$ for every other node $u\in V(T')$. Furthermore, $\tor_T(t_1)=\tor_{T'}(t_1)$ and, for the parent $p$ of $t$ (if it exists), we also have $\tor_T(p)=\tor_{T'}(p)$. It is then easily observed that the torso-size of all other nodes remains unchanged.

After exhaustively applying the contraction specified above, we arrive at a rooted tree-cut decomposition $(T'',\mc{X}'')$ of $G$ where each empty node has at least two children. Hence the number of empty nodes in $T''$ is upper-bounded by the number of non-empty nodes, which implies that $|T''|\leq 2|V(G)|$. $(T'',\mc{X}'')$ can then be transformed into a nice tree-cut decomposition without increasing the number of nodes by Lemma~\ref{lem:nicetcdec}.
\end{proof}

\section{FPT Algorithms}
\label{sec:algorithms}

In this section we will introduce a general dynamic programming framework for the design of FPT algorithms on nice tree-cut decompositions. The framework is based on leaf-to-root processing of the decompositions and can be divided into three basic steps:

\begin{enumerate}
\item[$\bullet$] {\bf Data Table}: definition of a data table $\cD_T(t)$ ($\cD(t)$ in brief) for a problem $\cP$ associated with each node $t$ of a nice tree-cut decomposition $(T,\mc{X})$.
\item[$\bullet$] {\bf Initialization and Termination}: computation of $\cD(t)$ in FPT time for any leaf $t$, and solution of $\cP$ in FPT time if $\cD(r)$ is known for the root $r$.
\item[$\bullet$] {\bf Inductive Step}: an FPT algorithm for computing $\cD(t)$ for any node $t$ when $\cD(t')$ is known for every child $t'$ of $t$.
\end{enumerate}
If the above holds for a problem $\cP$, then $\cP$ admits an FPT algorithm parameterized by tree-cut width. The {\bf Inductive Step} usually represents the most complex part of the algorithm. Our notion of nice tree-cut decompositions provides an important handle on this crucial step, as we formalize below.

Let $t$ be a node in a nice tree-cut decomposition. We use $B_t$ to denote the set of thin children $t'$ of $t$ such that $N(Y_{t'})\subseteq X_t$, and we let $A_t$ contain every child of $t$ not in $B_t$. The following lemma is a crucial component of our algorithms, since it bounds the number of children with nontrivial edge connections to other parts of the decomposition.

\begin{lemma}
\label{lem:Asmall}
Let $t$ be a node in a nice tree-cut decomposition of width $k$. Then $|A_t|\leq 2k+1$.
\end{lemma}

\begin{proof}
We partition the nodes in $A_t$ into two sets: $A'_t$ contains all thin nodes in $A_t$ and $A''_t$ contains all the bold nodes in $A_t$. We claim that $|A'_t|\leq k$ and $|A''_t|\leq k+1$, which will establish the statement. The inequality $|A'_t|\leq k$ is easy to see. Indeed, recall that $N(Y_b)\subseteq X_t \cup (V(G)\setminus Y_t)$ for every $b \in A'_t$  since $(T,\mc{X})$ is nice. Furthermore, each $b\in A'_t$ satisfies  $N(Y_{b})\cap (V(G)\setminus Y_t)\neq \emptyset$ since otherwise, $b$ would have been included in $B_t$. Therefore, each $b \in A'_t$ contributes at least one to the value $\adh(t)$. From $\adh(t)\leq k$, the inequality follows. 


To prove $|A''_t|\leq k+1$, suppose $|A''_t|=\ell\geq k+2$ for the sake of contradiction. Consider the torso $H_t$ at $t$. For each $b\in A_t \cup B_t$, let $z_b$ be the vertex of $H_t$ obtained by consolidating the vertex set $Y_b$ in $G$ and let $z_{top}$ be the vertex of $H_t$ obtained by consolidating the vertex set $V(G)\setminus Y_t$. Fix a sequence of suppressing vertices of degree at most two which yields a sequence of intermediate graphs $H_t=H_t^{(0)}, H_t^{(1)}, \cdots , H_t^{(m)}=\tilde{H}_t$. 
In this sequence, it is assumed that we suppress a vertex $z_b$ with $b\in A'_t\cup B_t$, namely which represents a thin node, whenever this is possible.

Observe that no vertex representing a bold node (i.e. a node of $A''_t$) is adjacent with a vertex representing a thin node (i.e. a node of $A'_t\cup B_t$) 
in any graph appearing in the suppression sequence. This observation is valid on $H_t^{(0)}=H_t$ because the tree-cut decomposition at hand is nice 
and thus in $H_t$ the vertex set $\{z_{top}\}\cup X_t$ separates the vertices obtained from $A'_t\cup B_t$ and those obtained from $A''_t$. By induction hypothesis, 
if a vertex $z_q$ representing a bold node $q$ becomes adjacent with a vertex $z_{q'}$ representing a thin node first time in the suppression sequence, say at $H_t^{(\ell)}$,
this is due to the suppression of $z_{top}$ in $H_t^{(\ell-1)}.$ Then, however, we should have suppressed $z_{q'}$ in $H_t^{(\ell)}$ before $z_{top}$ because $z_{q'}$ has degree at most two 
in $H_t^{(\ell)}$ (suppression does not increase the degree of any vertex) and we prioritize suppressing $z_{q'}$ over $z_{top}$. 

Let us choose $b',b'' \in A''_t$ so that $z_{b'}$ and $z_{b''}$ are the first and the second (distinct) vertices among $z_b$ for all $b\in A''_t$ whose degree strictly decreases in this sequence. Such $b'$ and $b''$ must exist since at least two vertices $z_b$, where $b\in A''_t$, do not appear in $\tilde{H}_t$, and any $z_b$ may only be removed by suppression. Let $a',a'' \in A_t\cup B_t\cup \{top\}$ and $0\leq i < j \leq m$ be such that the first decrease in the degrees of $z_{b'}$ and $z_{b''}$ is due to the suppressions of $z_{a'}\in V(H^{(i)}_t)$ and of $z_{a''}\in V(H^{(j)}_t)$. We observe that the respective degree of $z_{a'}$ and $z_{a''}$ are exactly one in $V(H^{(i)}_t)$ and in $V(H^{(j)}_t)$ since otherwise, the degree of $z_{b'}$ and $z_{b''}$ would not decrease. Moreover, $z_{a'}$ and $z_{b'}$ are adjacent in $V(H^{(i)}_t)$, and $z_{a''}$ and $z_{b''}$ are adjacent in $V(H^{(j)}_t)$


By the previous observation that $z_{b'}$ is non-adjacent with any vertex representing a thin node in $V(H^{(i)}_t)$ and $V(H^{(j)}_t)$, 
we conclude that $a'=top$. When $z_{top}$ is suppressed in $V(H^{(i)}_t)$ and accordingly decrease the degree of $z_{b'}$, 
this means that in $V(H^{(i)}_t)$, all the vertices obtained from thin nodes $A'_t\cup B_t$ have been eliminated in the previous steps. Therefore, all vertices in $V(H^{(i+1)}_t)\setminus X_t$ are those obtained 
from $A''_t$. It follows that  $a''=b'$. 

Notice that the suppressing of $z_{a'}$, namely $z_{top}$, decreases the degree of $z_{b'}$ by one. That is, the degree of $z_{b'}$ in $H^{(i+1)}_t$ remains at least two. Now that 
suppressing $z_{b'}$ in $H^{(j)}_t$ strictly decreases the degree of $z_{b''}$, the degree of $z_{b'}$ in $H^{(j)}$ equals to one. This implies that 
there is a vertex, say $z_{a^*}$, whose suppression further decreased the degree of $z_{b'}$ between the sequence of $H^{(i+1)}_t$ and $H^{(j)}_t$. 
However, then $a^*\in A''_t$, which contradicts our choice of $b''$ and $H^{(j)}_t$. This proves $|A''_t|=\ell\leq k+1$.
\end{proof}

In the remainder of this section we employ this high-level framework on the design of FPT algorithms parameterized by tree-cut width
for the following problems: \textsc{Capacitated Vertex Cover}, \textsc{Imbalance}, and \textsc{Capacitated Dominating Set}.

\subsection{Capacitated Vertex Cover}
\label{sub:cvc}
The \textsc{Capacitated Vertex Cover} is a generalization of the classical \textsc{Vertex Cover} problem
 which was originally introduced and studied in the classical setting \cite{ChuzhoyNaor02,GuhaHassinKhullerOr03}. More recently, its parameterized complexity has also been studied in combination with various parameters \cite{GuoNiedermeierWernicke07,DomLokstanovSaurabhVillanger08}. 
Unlike its uncapacitated variant, \textsc{Capacitated Vertex Cover} is known to be W[1]-hard when parameterized by treewidth \cite{DomLokstanovSaurabhVillanger08}. 

A capacitated graph is a graph $G=(V,E)$ together with a capacity function $c:V\rightarrow \Nat_0$. Then we call $C\subseteq V$ a \emph{capacitated vertex cover} of $G$ if there exists a mapping $f:E\rightarrow C$ which maps every edge to one of its endpoints so that the total number of edges mapped by~$f$ to any $v\in C$ does not exceed $c(v)$. We say that $f$ \emph{witnesses} $C$.

\begin{center}

  \vspace{0.3cm} \begin{boxedminipage}[t]{0.99\textwidth}
\begin{quote}
  tcw-\textsc{Capacitated Vertex Cover} (\cvc)\\ \nopagebreak
  \emph{Instance}: A capacitated graph $G$ on $n$ vertices together with a width-$k$ tree-cut decomposition $(T,\mc{X})$ of $G$, and an integer $d$.  \\ \nopagebreak
  \emph{Parameter}: $k$. \\ \nopagebreak
  \emph{Task}: Decide whether there exists a capacitated vertex cover $C$ of $G$ containing at most $d$ vertices. \nopagebreak
\end{quote}
\end{boxedminipage}
\end{center}
\subsubsection{Data Table, Initialization and Termination}

Informally, we store for each node $t$ two pieces of information: the ``cost'' of covering all edges inside $G[Y_t]$, and how much more it would cost to additionally cover edges incident to~$Y_t$. We formalize below.

For any graph $G=(V,E)$ and $U\subseteq V$, we let $\cvcv(G,U)$ denote the minimum cardinality of a capacitated vertex cover $C\subseteq U$ of $G$; if no such capacitated vertex cover exists, we instead let $\cvcv(G,U)=\infty$. For any node $t$ in a nice tree-cut decomposition of a capacitated graph $G=(V,E)$, we then use $a_t$ to denote $\cvcv(G[Y_t],Y_t)$.

Let $E_t$ denote the set of all edges with both endpoints in $Y_t$ and let $K_t$ denote the set of edges with exactly one endpoint in $Y_t$. Then $\cQ_t=\SB H=(Y_t\cup N(Y_t),E_t\cup E') \SM E'\subseteq K_t \SE$. Finally, we define $\beta_t: \cQ_t\rightarrow \Nat_0 \cup \{\infty\}$ such that 
for every $H$ in $\cQ_t$, $\beta_t(H)=\cvcv(H,Y_t)-a_t$ (whereas $\infty$ acts as an absorbing element and $\infty - \infty = \infty$). 

\begin{definition}
$\cD(t)=(a_t,\beta_t)$.
\end{definition}

Next, we show that the number of possible functions $\beta_t$ is bounded. To this end, we make use of the following lemma.

\begin{lemma}
\label{lem:cvcsizeaux}
Let $G=(V,E)$ be any capacitated graph and let $G'$ be obtained from $G$ by adding a pendant vertex $x$ (with arbitrary capacity). Let $U\subseteq V$. Then either $\cvcv(G',U)=\cvcv(G,U)$ or $\cvcv(G',U)=\cvcv(G,U)+1$ or $\cvcv(G',U)=\infty$.
\end{lemma}

\begin{proof}
We use induction on $|E|$. For $|E|=0$, the lemma is easily seen to be true, as $\cvcv(G,U)=0$ and $\cvcv(G',U)\in \{0,1,\infty\}$.

For the inductive step, let $\{z\}=N(x)$, $c=\cvcv(G,U)$ and $c'=\cvcv(G',U)$. If $c=\infty$ then clearly also $c'=\infty$, so we may assume that $c\neq \infty$; let $C\subseteq U$ be a capacitated vertex cover of $G$ of cardinality $c$. We distinguish the following cases:
\begin{enumerate}
\item if $z\not \in U$, then $c'=\infty$;
\item if $c(z)=0$, then $c'=\infty$;
\item otherwise, if $z\not \in C$, then $C'=C\cup \{z\}$ is a capacitated vertex cover of $G'$ and thus $c\leq c'\leq c+1$.
\end{enumerate}
The only case that remains is $z\in C$. It could occur that $C$ itself witnesses $c'=c$, i.e., it is possible to allocate edges to $C$ so that $z$ has free capacity for the new edge $xz$. If $z$ does not have free capacity to accommodate $xz$ within $C$ but $c'\neq \infty$, then there must exist a capacitated vertex cover $C'$ which ``frees up'' at least one edge which is mapped to $z$ in $C$; specifically, there must exist an edge $e=zq$ which is allocated to $z$ in $C$ and to $q$ in $C'$. Notice that if $q\not \in C$, we can add $q$ to $C$ and allocate $qz$ to $q$, hence allowing $xz$ to be allocated to $z$ and the lemma would hold. Thus we may assume $q\in C$.

We now construct a capacitated graph $H$ from $G$, the capacitated vertex cover $C$ and an arbitrary witness $f$ of $C$ as follows. $H$ is obtained from $G-z$ by attaching a new pendant vertex (with capacity $0$) to every $v\in V(H)$ such that $v\in N(z)$ and $f(vz)=v$. Since at least one neighbor of $z$, specifically $q$, will not receive a new pendant vertex in this construction, it follows that $|E(H)|<|E(G)|$.
Now we apply the inductive assumption on the capacitated graph $H$, the capacitated graph $H'$ obtained from $H$ by adding a pendant vertex $z'$ adjacent to $q$, and $U_H=U\setminus \{z\}$. Observe that $c=\cvcv(H,U_H)+1$ (since one can add $z$ to any capacitated vertex cover of size $\cvcv(H,U_H)$ and obtain a capacitated vertex cover for $G$). Furthermore, $C'\cup (N(z)\cap U)\setminus \{z\}$ witnesses $\cvcv(H',U_H)\neq \infty$, and so there are two final possibilities to consider:
\begin{enumerate}
\item $\cvcv(H',U_H)=\cvcv(H,U_H)$, in which case $c'=c$ since we can allocate the edge $zq$ to $q$ in $G$ without increasing the size of the vertex cover. 
\item $\cvcv(H',U_H)=\cvcv(H,U_H)+1$, in which case $c'=c+1=\cvcv(H',U_H)+1$ since we can allocate the edge $zq$ to $q$ in $G$ at a cost of $1$ additional vertex in the cover. 
\end{enumerate}
To conclude the proof, we make explicit the construction of a capacitated vertex cover demonstrating the above claims. We start with a capacitated vertex cover $K$ of $H'$ witnessing $\cvcv(H',U_H)=\cvcv(H,U_H)$ (or $\cvcv(H',U_H)=\cvcv(H,U_H)+1$) and add to it the vertex $z$. As for the allocation of edges, we retain the same allocation of all edges not incident to $z$ as in $K$, allocate $zq$ to $q$ (this is made possible by the construction of $H'$) and use $f$ to allocate all the remaining edges incident to $z$ (this is made possible by the pendant vertices added during the construction of $H$). By ``saving'' on the edge $zq$, it is now possible to accommodate the new edge $xz$ into $z$.
\end{proof}

\begin{lemma}
\label{lem:cvcsize}
Let $k$ be the width of a nice tree-cut decomposition $(T,\mc{X})$ of $G$ and let~$t$ be any node of $T$. Then $\beta_t(H)\in [k]\cup \{\infty\}$ for every $H\in \cQ_t$.
\end{lemma}

\begin{proof}
We actually prove a slightly stronger claim: that for every $E'\subseteq K_t$, $\beta_t(H)\in [|E'|]\cup {\infty}$. Notice that $|K_t|=\adh(t)$ and hence $|E'|\leq k$. We proceed by induction; for $|E'|=0$, the claim holds by definition. 

For the inductive step, let $E_1'=E_0'\cup \{e\}$ and $e=ab$ where $a$ is a vertex in $Y_t$ and $b$ is a neighbor of $a$ in $N(Y_t)$, 
and let $H_0$ and $H_1$ be the graphs  $(Y_t\cup N(Y_t), E_t\cup E'_0)$ and $(Y_t\cup N(Y_t), E_t\cup E'_1)$ respectively. 
  The inductive claim then follows directly from Lemma~\ref{lem:cvcsizeaux} by making the pendant vertex $x$ adjacent to $a$ in $H_0$; indeed, observe that adding an edge between $a$ and $b$ has precisely the same effect on the capacitated vertex set as adding a new pendant vertex to $a$ (since we restrict the capacitated vertex set to $Y_t$).
\end{proof}

Since the graphs $H\in \cQ_t$ that appear in the domain of $\beta_t$ can be uniquely represented by the choice of $E'\subseteq K_t$, we can store $\beta_t$ as a subset of $2^{K_t}\times ([k]\cup \{\infty\})$.

%

\subsubsection{Initialization and Termination}
\begin{lemma}
\label{lem:cvcinit}
Let $t$ be a leaf in a nice tree-cut decomposition $(T,\mc{X})$ of a capacitated graph $G$, and let $k$ be the width of $(T,\mc{X})$. Then $\cD(t)$ can be computed in time $2^{\bigoh(k\cdot \log k)}$.
\end{lemma}

\begin{proof}
$a_t=\cvcv(G[Y_t],Y_t)$ can be computed in time $2^{\bigoh(k\cdot \log k)}$ by known results~\cite{DomLokstanovSaurabhVillanger08}, since $|Y_t|\leq k$. Furthermore, $|K_t|\leq k$ since $\adh(t) \leq k$, and so $|\cQ_t|\leq 2^k$. What remains is to compute $\beta_t$ by determining $\cvcv(H,Y_t)$ for each $H$ in $\cQ_t$; for each $H$ this can be done by using the same $2^{\bigoh(k\cdot \log k)}$ algorithm for capacitated vertex cover while setting the capacities of vertices outside of $Y_t$ to $0$. In this way we can compute $\beta_t$ in time $2^k\cdot 2^{\bigoh(k\cdot \log k)}$, from which the lemma follows.
\end{proof}

\begin{observation}
\label{obs:cvcend}
Let $(G,d)$ be an instance of \cvc~and let $r$ be the root of a nice tree-cut decomposition of $G$. Then $(G,d)$ is a yes-instance if and only if $a_r\leq d$.
\end{observation}

\subsubsection{Inductive Step}

Our next and final goal is to show how to compute $\cD(t)$ of a node~$t$ once we have $\cD(t')$ for  each child $t'$ of $t$. We formalize this problem below.

\begin{center}
\vspace{0.3cm} 
\begin{boxedminipage}[t]{0.99\textwidth}
\begin{quote}
  \textsc{CVC Join}\\ \nopagebreak
  \emph{Instance}: A \CVC~instance consisting of a capacitated graph $G$ and an integer $d$, a non-leaf node $t$ of a width-$k$ nice tree-cut decomposition $(T,\mc{X})$ of $G$, and $\cD(t')$ for each child $t'$ of $t$.\\
    \emph{Parameter:} $k$. \\ \nopagebreak
  \emph{Task}: Compute $\cD(t)$. \nopagebreak
\end{quote}
\end{boxedminipage}
\end{center}

We use a two-step approach to solve \textsc{CVC Join}. First, we reduce the problem to a simplified version, which we call \textsc{Reduced CVC Join} and which has the following properties: $A_t$ is empty, $\adh(t)=0$, and $G[X_t]$ is edgeless. Recall that $B_t$ is the set of all thin children $t'$ of node $t$ with $N(Y_{t'})\subseteq X_t$ and $A_t$ denotes the set of 
the  children of $t$ not in $B_t$. 

For the following lemma, we remark that the linear dependency on $n$ is merely caused
by the fact that the instances of \textsc{Reduced CVC Join} obtained in the Turing reduction are of size $\bigoh(n)$. It is conceivable that this linear factor could be avoided with the use of a carefully designed data structure.

\begin{lemma}
\label{lem:cvcreduce}
There is an FPT Turing reduction from \textsc{CVC Join} to $2^{\bigoh(k^2)}$ instances of \textsc{Reduced CVC Join} which runs in time $2^{\bigoh(k^2)}\cdot n$. 
\end{lemma}

\begin{proof}
Let $E_x=\SB ab\in E \SM a,b\in X_t \SE$ and $E_1=\SB ab\in E \SM \exists t_a\in A_t: a\in Y_{t_a}, b\in (Y_t\setminus Y_{t_a})\SE$; in other words, $E_1$ contains edges in $Y_t$ which contribute to the adhesion of some node of $A_t$. Let $E'$ denote the set of edges with only one endpoint in $Y_t$. By Lemma~\ref{lem:Asmall} and the bound of $k$ on the width of $(T,\mc{X})$, it follows that $|E_1|\leq 2k^2+k$ and $|E_x|\leq k^2$. Let $\cF$ contain all mappings of edges from $E_x\cup E_1\cup E'$ to one of their endpoints; formally $\cF=\SB f:(E_x\cup E_1\cup E')\rightarrow V\SM f(ab)=a\text{ or } f(ab)=b\SE$. Notice that $|\cF|\leq 2^{\bigoh(k^2)}$.

    For any $f\in \cF$ and $t'\in A_t\cup \{t\}$, we define the \emph{completion} $U_{f,t'}$ as the graph $G[Y_{t'}]$ together with the vertices $N(Y_{t'})$ and each edge $e$ such that $f(e)\in Y_{t'}$; informally, $f$ designates how we want to cover our edges, and the completion is the relevant subproblem of covering edges in such a way within $Y_{t'}$. Let $X(f)=\SB v\in X_t \SM \exists e: f(e)=v \SE$.
For each $f$, we compute $\acost(f)=\sum_{t'\in A_t}\beta_{t'}(U_{f,t'})+a_{t'}$; informally, $\acost$ contains the minimum cost of a capacitated vertex cover in $A_t$ complying with $f$ (obtained from information stored in the data tables for $A_t$), while $X(f)$ contains vertices of $X_t$ which must be in the capacitated vertex cover for this choice of $f$.

We now construct an instance $I_f$ of \textsc{Reduced CVC Join} for each $f$, as follows.
  We remove all nodes in $A_t$, remove all edges with an endpoint outside of $Y_t$, remove all edges with both endpoints in $X_t$. We add a single child $t''$ with $\adh(t'')=0$ and $\cD(t'')=\{\acost(f),\{G[Y_{t''}]\mapsto 0\}\}$ (to carry over information of the cost of all forgotten nodes in $A_t$), and for each $v\in X(f)$ we increase the capacity of $v$ by $1$ and add a pendant vertex $v'$ adjacent to $v$ with capacity $0$ (ensuring that $v$ must be taken into the vertex cover in $I_f$). Then for each $v\in X(f)$ we reduce its capacity by the number of edges mapped to $v$ by $f$ (if this would reduce its capacity to negative values, we set its capacity to $0$ and attach a $0$-capacity pendant vertex to $v$; this ensures the non-existence of a capacitated vertex cover for $I_f$). We adapt the tree-cut decomposition accordingly, by adding the new pendant vertices into new bags in $B_t$. Notice that $I_f$ now has the desired properties of \textsc{Reduced CVC Join}. 

Since its adhesion is zero, the solution of $I_f$ contains a single tuple $(a_{I_f},\{U_{f,t}\mapsto 0\})$. It remains to show how one can use this tuple to obtain the records for the original instance of \textsc{CVC Join}. For each graph $H$ in $\cQ_t$, we apply brute-force enumeration over $\cF$ to compute the set $\cF_H$ of all elements of $\cF$ such that $U_{f,t}\cong H$ (intuitively, this corresponds to identifying all the ways for covering the missing edges). We proceed by selecting a function $f\in \cF_H$ such that $a_{I_f}$ is minimized, and denote this $a_{I_f}$ by $\cost(H)$. For $H=G[Y_t]$ we then set $a_t=\cost(H)$, and we construct $\beta_t$ by setting, for each $H$ in $\cQ_t$ and each $U_{f,t}$ , $\beta_t:H\mapsto (\cost(H)-a_t)$.

To argue the claimed running time, observe that one can compute the values of $\acost(f)$ for all $f\in \cF$ in time $2^{\bigoh(k^2)}$, and that the instance $I_f$ for each $f\in \cF$ obtained by trivially modifying the graph and the tree-cut decomposition can be constructed in $\bigoh(n)$ time. 
%

  
We argue the correctness of our reduction by arguing the correctness of the computed value $a_t$; the same argument then applies analogously also to the correctness of the construction of $\beta_t$. Assume that the value $a_t$ computed above is greater than $\cvcv(G[Y_t],Y_t)=a'$. Then there exists a capacitated vertex cover $C$ of $G[Y_t]$ of cardinality $a'$ and a mapping of edges $f_C$ which witnesses $C$. Let $f\in \cF$ be the (unique) restriction of $f_C$ to the set of edges mapped by functions in $\cF$, and for each $t'\in A_t$ we let $C_{t'}$ be the subset of $C$ which intersects $Y_{t'}$. By the correctness of the data tables for $A_t$, it holds that $\sum_{t'\in A_t} |C_{t'}|=\acost(f)$. Furthermore, each $v\in X(f)$ must occur in $C$ by $f_C$, and also in any solution to the instance $I_f$ due to the added pendant vertex. We would hence necessarily conclude that $a_{I_f}>C\setminus (X(f)\cup \bigcup_{t'\in A_t} C_{t'})$, which however contradicts the correctness of the solution to $I_f$.

On the other hand, assume that $a_t<a'$. Then we can construct a capacitated vertex cover $C$ of $G[Y_t]$ of cardinality $a_t$ from $X(f), f,$ and the partial vertex covers which witness the correctness of the solution to $I_f$.
\end{proof}


\begin{lemma}
\label{lem:redcvcfpt}
There exists an algorithm which solves \textsc{Reduced CVC Join} in time $k^{\bigoh(k^2)}\cdot (|B_t|+1)$.
\end{lemma}

\begin{proof}
Since $K_t=\emptyset$, it suffices to compute $\cvcv(G[Y_t],Y_t)$ and output $\cD(t)=\{\cvcv(G[Y_t],Y_t),G[Y_t]\mapsto 0\}$. We branch over all the at most $2^k$ possible sets $X'\subseteq X_t$, representing possible intersections of the capacitated vertex cover with $X_t$. For each such $X'$, we construct an ILP formulation which computes the minimum capacitated vertex cover in $G[Y_t]$ which intersects with $X_t$ in $X'$. We begin by having a constant $c_x$ contain the capacity of $x$ for every $x \in X'$, and let $c_x=0$ for $x \in X_t\setminus X'$.

We define a relation $\equiv$ on $B_t$. For $t_1,t_2\in B_t$, we say $t_1\equiv t_2$ if (i) $N(Y_{t_1})=N(Y_{t_2})$, and (ii)  there exists a bijection $\phi$ from $\cQ_{t_1}$ to $\cQ_{t_2}$ which satisfies the following two conditions for every $H_1\in \cQ_{t_1}$ and $H_2\in \cQ_{t_2}$ with $H_2=\phi(H_1)$.
\begin{enumerate}
\item $deg_{H_1}(x)=deg_{H_2}(x)$ for every $x\in N(Y_{t_1})$
\item $\beta_{t_1}(H_1)=\beta_{t_2}(H_2)$.
\end{enumerate}

It is easy to see that $\equiv$ defines an equivalence relation. We denote by $[\equiv]$ the set of equivalence classes of $\equiv$, and note that $|[\equiv]|\leq \bigoh(k^2)$. Indeed, there are only $\bigoh(k^2)$ subsets of $X_t$ containing at most two vertices. Moreover, Lemma~\ref{lem:cvcsizeaux} guarantees $\beta_t(H) \in [2]\cup \{\infty\}$ for $H\in \cQ_t$ for any $t \in B_t$. Hence, the number of ways that $H\in \cQ_t$ can be mapped to $\Nat_0$ under $\beta_t$ is bounded by a constant and thus, $|[\equiv]|\leq \bigoh(k^2)$.

Given an equivalence class $\nu \in [\equiv]$, let us fix a node $t_\nu\in \nu$ and let $H_\nu$ be the graph $(Y_{t_\nu}\cup N(Y_{t_\nu}), E_{t_\nu}\cup K_{t_\nu})$. For $H\in \cQ_{t_\nu}$, we create a variable $z_H$. This variable shall denote the number of nodes $t'$ in $\nu$ such that among the edges in $K_{t'}$, exactly $deg_H(x)$ many edges can be covered by $Y_{t'}$  for each $x \in N(Y_{t'})$ while incurring the cost $\beta_{t_\nu}(H)$. Our ILP instance is as follows, for which an informal explanation is added in the subsequent paragraphs.
\begin{equation*}
\begin{array}{ll@{}ll}
\text{minimize}  & \displaystyle \sum_{H \in \cQ_{t_\nu}, t_\nu\in \nu, \nu\in [\equiv]} \beta_{t_\nu}(H)\cdot z_H+|X'| &\\
\text{subject to}& \displaystyle \sum_{H \in \cQ_{t_\nu}, t_\nu\in \nu} z_H= |\nu|  &\forall \nu\in [\equiv]\\
			& \displaystyle \sum_{H\in \cQ_{t_\nu}} (deg_{H_\nu}(x)-deg_H(x))\cdot z_H \leq c_x  \qquad &\forall x\in X_t
\end{array}
\end{equation*}
where $z_H$ is a non-negative integer variable for every $H\in \cQ_{t_\nu}$, $\nu\in [\equiv]$.

%
Notice that due to the definition of $\equiv$, any two nodes $t_1,t_2\in \nu$ covering the same number of edges incident with $x \in N(Y_{t_\nu})$ will incur the same cost. This justifies expressing the cost $\cvcv(G[Y_t],Y_t)$ with an objective function as above. Since $|\cQ_{t_\nu}|\leq 4$ and we create the variables $z_H$ for each $H \in \cQ_{t_\nu}$ for each equivalence class $\nu\in [\equiv]$, the total number of integer variables created is in $\bigoh(k^2)$ and so an optimal solution can be obtained by Theorem~\ref{thm:pilp} in time $k^{\bigoh(k^2)}\cdot |B_t|$.

Overall, the formulation relies on grouping interchangeable children into equivalence classes, and uses variables to denote the number of children in each class for which we use the same assignment of edges between $X_t$ and $Y_{t'}$. For instance, if there is only a single edge to a specific $x\in X_t$, variable $z_{H_\alpha}$ (with $H_\alpha$ being the graph containing the edge incident to $x$) contains the number of children where the edge is mapped to $x$ while $z_{H_\beta}$ (with $H_\beta$ being the graph without the edge incident to $x$) contains the number of children where this edge is mapped to $Y_{t'}$. Constraint $1$ ensures that the variables indeed capture a partition of children in $i$, while constraint $2$ ensures that the capacities in $X_t$ are not exceeded. The instance minimizes the increase in cardinality of capacitated vertex covers over all $Y_{t'}$ by multiplying the number of children with their respective value of $\beta$ based on the chosen assignment of edges.

Let $\cI$ be the minimum value of  $\sum_{H \in \cQ_{t_\nu}, \nu\in [\equiv]} \beta_{t_\nu}(H)\cdot z_H+|X'|$ over all ILP instances described above; if no solution exists or 
if a solution is larger than the preset (large) constant representing $\beta_{t'}(H)=\infty$, then we set $\cI=\infty$. We argue that $\cvcv(G[Y_t],Y_t)=\sum_{t'\in B_t}a_{t'}+\cI$, from which the proof follows.


First, consider for a contradiction that $\cvcv(G[Y_t],Y_t)<\sum_{t'\in B_t}a_{t'}+\cI$. Then there exists a capacitated vertex cover in $G[Y_t]$ of cardinality $\cvcv(G[Y_t])$. Let $X'$ be its intersection with $X_t$ and let $f$ be its witness function. Let $U_{f,t'}$ be the graph on the vertex set $Y_{t'} \cup N(Y_{t'})$ with an edge set $E_{t'} \cup \{e\in K_{t'}: f(e)\in Y_{t'} \}$. Note that $U_{f,t'}\in \cQ_{t'}$.
For each equivalence class $\nu\in [\equiv]$, we again partition $\nu$ into sets so that $t_1,t_2\in \nu$ belong to the same set if and only if $\phi(U_{f,t_1})=U_{f,t_2}$. Here, $\phi$ is the bijection between any two nodes $t_1,t_2 \in \nu$ ensuring the equivalence of $t_1$ and $t_2$. By fixing a node $t_\nu\in \nu$, each set can be uniquely 'represented' by some $H\in \cQ_{t_\nu}$ due to the bijective mapping. For each $H\in \cQ_{t_\nu}$, let $\tilde{z}_H$ denote the cardinality of the set represented by $H$. One can check that $z_H=\tilde{z}_H$, $\forall H\in \cQ_{t_\nu}, \forall \nu\in [\equiv]$, is a feasible solution for the above ILP yielding an objective value $\cvcv(G[Y_t],Y_t)$, contradicting the optimality of $\cI$. 

On the other hand, consider the case $\cvcv(G[Y_t],Y_t)>\sum_{t'\in B_t}a_{t'}+\cI$. Then one could construct a better capacitated vertex cover for $G[Y_t]$, by having a capacitated vertex cover intersect $X_t$ in the $X'$ used to obtain $\cI$ and using a witness function $f$ which maps edges between $X_t$ and $B_t$ based on the partition of the equivalence classes discovered by the ILP formulation in order to reach $\cI$. Again, one can check that the resulting capacitated vertex cover would have cardinality exactly $\sum_{H \in \cQ_{t_\nu}, \nu\in [\equiv]} \beta(H)\cdot z_H+|X'|$.
\end{proof}

By combining Lemma~\ref{lem:cvcreduce} with Lemma~\ref{lem:redcvcfpt}, we obtain:

\begin{corollary}
\label{cor:cvcjoin}
There exists an algorithm which solves \textsc{CVC Join} in time $k^{\bigoh(k^2)}\cdot n$.
\end{corollary}

\begin{theorem}
\label{thm:CVCalg}
\CVC~can be solved in time $k^{\bigoh(k^2)}\cdot n^2$.
\end{theorem}

\begin{proof}
We use Theorem~\ref{thm:decompsize} to transform $(T,\mc{X})$ into a nice tree-cut decomposition with at most $2n$ nodes. We then use Lemma~\ref{lem:cvcinit} to compute $\cD(t)$ for each leaf $t$ of $T$, and proceed by computing $\cD(t)$ for nodes in a bottom-to-top fashion by Corollary~\ref{cor:cvcjoin}. 
The total running time per node is at most $n\cdot k^{\bigoh(k^2)}$, and there are $\bigoh(n)$ many nodes. 
Once we obtain $\cD(r)$, we can correctly output by Observation~\ref{obs:cvcend}. 
\end{proof}

\subsection{Imbalance}
\label{sub:imb}

The \textsc{Imbalance} problem was introduced by Biedl et
al.~\cite{BiedlChanGHW05}. It was shown to be equivalent to \textsc{Graph Cleaning} \cite{GaspersMessingerNP09}, and was studied in the parameterized setting  \cite{LokshtanovMisraSaurabh13,FellowsLokshtanovMisraRS08}. 
The problem is FPT when parameterized
by $\mathbf{degtw}$~\cite{LokshtanovMisraSaurabh13}. In this subsection we
prove that \textsc{Imbalance} remains FPT even when parameterized by
the more general tree-cut width.

  Given a linear order $R$ of vertices in a graph $G$, let $\lhd_R(v)$ and $\rhd_R(v)$ denote the number of neighbors of $v$ which occur, respectively, before (``to the left of'') $v$ in $R$ and after (``to the right of'') $v$ in $R$. The \emph{imbalance} of a vertex $v$, denoted $\imbv_R(v)$, is then defined as the absolute value of $\rhd_R(v)-\lhd_R(v)$, and the imbalance of $R$, denoted $\imbv_R$, is equal to $\sum_{v\in V(G)}\imbv_R(v)$.

\begin{center}
  \vspace{0.3cm} 
  \begin{boxedminipage}[t]{0.99\textwidth}
\begin{quote}
  tcw-\textsc{Imbalance} (\imb)\\ \nopagebreak
  \emph{Instance}: An $n$-vertex graph $G=(V,E)$ with $|V|=n$, and a width-$k$ tree-cut decomposition $(T,\mc{X})$ of $G$, and an integer $d$.  \\ \nopagebreak
  \emph{Parameter}: $k$. \\ \nopagebreak
  \emph{Task}: Decide whether there exists a linear order $R$ of $V$ such that $\imbv_R\leq d$.
\nopagebreak
\end{quote}
\end{boxedminipage}
\end{center}
\subsubsection{Data Table}

Let $A\subseteq B$ be sets and let $f_A, f_B$ be linear orders of $A,B$ respectively. We say that $f_A$ is a \emph{linear suborder} of $f_B$ if the elements of $A$ occur in the same order in $f_A$ as in $f_B$ (similarly, $f_B$ is a \emph{linear superorder} of $f_A$). The information we remember in our data tables can be informally summarized as follows. First, we remember the minimum imbalance which can be achieved by any linear order in $Y_t$. Second, for each linear order $f$ of vertices which have neighbors outside of $Y_t$ \emph{and} for a restriction on the imbalance on these vertices, we remember how much the imbalance grows when considering only linear superorders of $f$ which satisfy these restrictions. The crucial ingredient is that the restrictions mentioned above are ``weak'' and we only care about linear superorders of $f$ which do not increase over the optimum ``too much''; this allows the second, crucial part of our data tables to remain bounded in $k$.

For brevity, for $v\in Y_t$ we let $\adh_t(v)$ denote $|N_{V\setminus Y_t}(v)|$, i.e., the number of neighbors of $v$ outside $Y_t$. Let $f$ be a linear order of $\partial(Y_t)$ and let $\tau$ be a mapping such that $\tau(v)\in \{-\infty,-\adh_t(v),-\adh_t(v)+1,\dots,\adh_t(v),\infty \}$ for every $v\in \partial(Y_t)$. We then call a tuple of the form $(f,\tau)$ an \emph{extract} (of $Y_t$ or simply put, of $t$), and let $\cL_t$ denote the set of all extracts (for nodes with adhesion at most $k$) of $t$; when the node $t$ is clear from the context, we will use $\cL$ as shorthand for $\cL_t$. The extract $\alpha=(f,\tau)$ is realized in $Y_t$ (by $R$) if there exists a linear order $R$ of $Y_t$ such that
\begin{enumerate}
\item $R$ is a linear superorder of $f$, and
\item for each $v\in \partial(Y_t)$:
\begin{itemize}
\item if $\tau(v)\in \mathbb{Z}$ then $\imbv_R(v)=\tau(v)$,
\item if $\tau(v)=-\infty$ then $\rhd_R(v)-\lhd_R(v)\leq -\adh_t(v)-1$,
\item if $\tau(v)=\infty$ then $\rhd_R(v)-\lhd_R(v)\geq \adh_t(v)+1$.
\end{itemize}
\end{enumerate}

The cost of a realized extract $\alpha$, denoted $c(\alpha)$, is the minimum value of $\sum_{v\in Y_t}\imbv_R(v)$ over all $R$ which realize $\alpha$ (notice that edges with only one endpoint in $Y_t$ do not contribute to $c(\alpha)$). If $\alpha$ is not realized in $Y_t$, we let $c(\alpha)=\infty$. We store the following information in our data table: the cost of a minimum extract realized in $Y_t$, and the cost of every extract whose cost is not much larger than the minimum cost. We formalize below; let $e_t$ denote the number of edges with one endpoint in $Y_t$.

\begin{definition}
$\cD(t)=(a_t,\beta_t)$ where $a_t=\min_{\alpha\in \cL}{c(\alpha})$ and $\beta_t:\cL\rightarrow \Nat_0\cup \{\infty\}$ such that $\beta_t(\alpha)=c(\alpha)-a_t$ if $c(\alpha)-a_t\leq 4e_t$ and $\beta_t(\alpha)=\infty$ otherwise.
\end{definition}

Notice that we are deliberately discarding information about the cost of extracts whose cost exceeds the optimum by over $4e_t$. We justify this below.

\begin{lemma}
\label{lem:imbdiscard}
Let $G=(V,E)$ be a graph and let $t$ be a node in a width-$k$ tree-cut decomposition $(T,\mc{X})$ of $G$. Let $\alpha_1, \alpha_2$ be two extracts of $Y_t$ such that $c(\alpha_1)>4e_t+c(\alpha_2)$. Then for any linear order $R_1$ of $Y_t$ which realizes $\alpha_1$ and any linear superorder $R$ (over $V$) of $R_1$, it holds that there exists a linear order $R'$ over $V$ such that $\imbv_{R'}<\imbv_R$. 
\end{lemma}

\begin{proof}
Consider any linear order $R_2$ (over $Y_t$) which realizes $\alpha_2$ of cost $c(\alpha_2)$; since $c(\alpha_2)\neq \infty$, there must exist at least one such $R_2$. Let $R$ be any linear order constructed above; clearly, there exists a unique linear order $R^*$ over $V\setminus Y_t$ which is a linear suborder of $R$. Consider the linear order $R'$ obtained by the concatenation of $R_2$ and $R^*$. Let $I^*$ denote the imbalance of $R^*$ in $G[V\setminus Y_t]$, and recall that $c(\alpha_2)$ denotes the imbalance of $R'$ in $G[Y_t]$. Since the addition of $e_t$ edges can only increase or decrease the imbalance by at most $2e_t$, we obtain that $\imbv_{R'}\leq I^*+c(\alpha_2)+2e_t$ and $\imbv_R\geq I^*+c(\alpha_1)-2e_t$. 
Altogether we conclude $\imbv_{R'}\leq I^*+c(\alpha_2)+2e_t< I^* + c(\alpha_1)-2e_t\leq \imbv_R$.
\end{proof}

The following observation establishes a bound on the size of our records. Moreover, it will be useful to note that for a fixed value of the adhesion, one can efficiently enumerate all the functions that may appear in the records (this will facilitate the processing of thin nodes).

\begin{observation}
\label{obs:imbsmall}
The cardinality of $\cL$ is bounded by $k^{\bigoh(k)}$. Moreover, for each node $t$ with adhesion $\kappa$, the number of possible functions $\beta_t$ is bounded by $\kappa^{\kappa^{\bigoh(\kappa)}}$, and these may be enumerated in the same time.
\end{observation}

\subsubsection{Initialization and Termination}
\begin{lemma}
\label{lem:imbinit}
Let $t$ be a leaf in a nice tree-cut decomposition $(T,\mc{X})$ of a graph $G$, and let $k$ be the width of $(T,\mc{X})$. Then $\cD(t)$ can be computed in time $k^{\bigoh(k)}$.
\end{lemma}

\begin{proof}
We can branch over all at most $k^k$ linear orders of $Y_t$, and for each we can compute its imbalance in $G[Y_t]$ in time $\bigoh(k)$. This already gives sufficient information to construct $\cD(t)$.
\end{proof}

\begin{observation}
\label{obs:imbend}
Let $(G,d)$ be an instance of \imb~and let $r$ be the root of a nice tree-cut decomposition of $G$. Then $(G,d)$ is a yes-instance if and only if $a_r\leq d$.
\end{observation}

\subsubsection{Inductive Step}
What remains is to show how to compute $\cD(t)$ for a node $t$ once $\cD(t')$ is known for each child $t'$ of $t$. We formalize this problem below.

\begin{center}
  \vspace{0.3cm} 
  \begin{boxedminipage}[t]{0.99\textwidth}
\begin{quote}
  \textsc{IMB Join}\\ \nopagebreak
  \emph{Instance}: A \imb~instance $(G,d)$, a non-leaf node $t$ of a width-$k$ nice tree-cut decomposition $(T,\mc{X})$ of $G$, and $\cD(t')$ for each child $t'$ of $t$.\\
    \emph{Parameter:} $k$. \\ \nopagebreak
  \emph{Task}: Compute $\cD(t)$. \nopagebreak
\end{quote}
\end{boxedminipage}
\end{center}

Once again, we use the two-step approach of first reducing to a
``simpler'' problem and then applying a suitable ILP encoding. 
We call the problem we reduce to \textsc{Reduced IMB Join}.

\begin{center}
  \vspace{0.3cm} 
  \begin{boxedminipage}[t]{0.99\textwidth}
\begin{quote}
  \textsc{Reduced IMB Join}\\ \nopagebreak
  \emph{Instance}: A \imb~instance consisting of a graph $G$ and an integer $d$, a root $t$ of a nice tree-cut decomposition $(T,\mc{X})$ of $G$ such that $A_t=\emptyset$ and $G[X_t]$ is edgeless, $\cD(t')$ for each child $t'$ of $t$, a linear order $f$ of $X_t$, a mapping $\omega:X_t\rightarrow \mathbb{Z}$, and a set $\zeta$ of linear constraints on the value of $\rhd(v)-\lhd(v)$ for a subset of $X_t$.\\
    \emph{Parameter:} $k=|X_t|$. \\ \nopagebreak
  \emph{Task}: Compute the minimum value of $\big(\sum_{v\in Y_t\setminus X_t}\imbv_R(v)\big)+\big(\sum_{v\in X_t} |\rhd_R(v)-\lhd_R(v)+\omega(v)|\big)$ over all linear superorders $R$ of $f$ over $Y_t$ which satisfy $\zeta$, or correctly determine that no such $R$ exists. \nopagebreak
\end{quote}
\end{boxedminipage}
\end{center}

\begin{lemma}
\label{lem:imbreduce}
There is an FPT Turing reduction from \textsc{IMB Join} to $k^{\bigoh(k^2)}$ instances of \textsc{Reduced IMB Join} which runs in time $k^{\bigoh(k^2)}\cdot n$. 
\end{lemma}

\begin{proof}
Our goal is to use the data contained in each $\cD(t')$ to preprocess all information in the nodes $A_t$, leaving us with a set of \textsc{Reduced IMB Join} instances.
We branch over the at most $k^{\bigoh(k)}$ extracts (of $t$) in $\cL$ . The remainder of the proof shows how to compute $c(\alpha)$ for each extract $\alpha=(f,\tau)\in \cL$ using an oracle which solves \textsc{Reduced IMB Join}. Indeed, it follows from the definition of $\cD(t)$ that this information is sufficient to output $\cD(t)$ (in $\bigoh(k)$ time).

Let $Z=\SB v\in Y_t\SM v\in X_t \text{ or } \exists t'\in A_t: v\in \partial(Y_{t'}) \SE$; by Lemma~\ref{lem:Asmall} and the bound on the width of $(T,\mc{X})$, it follows that $|Z|\leq 2k^2+2k$. Let $\cJ$ then denote the set containing all the at most $k^{\bigoh(k^2)}$ linear orders of $Z$. Next, we prune $\cJ$ to make sure it is compatible with our $\alpha$; specifically, we prune $\cJ$ by removing any linear order which is not a linear superorder of $f$. We now branch over $\cJ$; let $j$ be a fixed linear order in~$\cJ$.

For each $t'\in A_t$ we compute the set $\delta(t')$ of extracts (of $t'$) which are ``compatible'' with our chosen extract $\alpha=(f,\tau)$ (of $t$) and our chosen $j$. If $Y_{t'}\cap \partial(Y_t)=\emptyset$, then $\delta(t')=\cL$. However, if $Y_{t'}\cap \partial(Y_t)=:Q\neq \emptyset$ then we let $\delta(t')$ contain an extract $(f',\tau')\in \cL$ if and only if the following holds. 
For every $v\in Q$ such that $\tau(v)=\infty$, it holds that either $\tau'(v)=\infty$ or $\tau'(v)+\rhd_j(v)-\lhd_j(v)\geq \adh_t(v)+1$. For every $v\in Q:\tau(v)=-\infty$ it holds that either $\tau'(v)=-\infty$ or $\tau'(v)+\rhd_j(v)-\lhd_j(v)\leq -\adh_t(v)-1$. For every $v\in Q$ such that $\tau(v)=i$ where $i\neq \infty$, it holds that $i=\tau'(v)+\rhd_j(v)-\lhd_j(v)$.

Finally, we branch over all choices of compatible extracts for each $t'\in A_t$. Specifically, we branch over all choices of $\lambda$, where $\lambda(t')=\alpha'\in \delta(t')$ for 
every $t'\in A_t$. For each $\lambda$, we compute the total imbalance $\imbv_\lambda$ in $A_t$, specifically the sum of the imbalances of vertices in $Y_{t'}$ for every $t'\in A_t$ excluding edges which have an endpoint outside of $Y_t$. This can be computed as follows: $\imbv_\lambda=\sum_{t'\in A_t}(a_{t'}+\beta(\lambda(t')))+K$, where $K$ is the total adjustment of the imbalance caused by edges between $X_t$ and some $Y_{t'}$ (note that $K$ may be negative). $K$ can be computed by the following simple procedure. Initialize with $K:=0$ and $\tau'':=\tau'$, and for each edge $ab$ where $a\in X_t$ and $b\in Y_{t'}$ such that $\lambda(t')=(f',\tau')$:
\begin{itemize}
\item  if $\tau''(b)>0$ and $j(a)<j(b)$ then $K:=K-1$ and $\tau''(b):=\tau''(b)-1$;
\item  if $\tau''(b)\geq 0$ and $j(b)<j(a)$ then $K:=K+1$ and $\tau''(b):=\tau''(b)+1$;
\item  if $\tau''(b)\leq 0$ and $j(a)<j(b)$ then $K:=K+1$ and $\tau''(b):=\tau''(b)-1$;
\item  if $\tau''(b)< 0$ and $j(b)<j(a)$ then $K:=K-1$ and $\tau''(b):=\tau''(b)+1$.
\end{itemize}

The same procedure as the one above can be used to also compute the value $\rhd_j(v)-\lhd_j(v)$ for each $v\in X_t$, which we store as $\omega(v)$. From $\tau$, we obtain a set of constraints $\zeta$ which will guarantee that the solution to \textsc{Reduced IMB Join} will be compatible with $\alpha$. Specifically, for each $v\in X_t\cap \partial(Y_t)$ such that $\tau(v)=\infty$, we set $\zeta(v)\geq \adh_t(v)+1$; if $\tau(v)=-\infty$, then we set $\zeta(v)\leq -\adh_t(v)-1$; and if $\tau(v)\in \mathbb{Z}$, then we set $\zeta(v)=\tau(v)$. As for $f$, we obtain it as the unique linear suborder of $j$ restricted to $X_t$. This completes the construction of our \textsc{Reduced IMB Join} instance $I_{\alpha,j,\lambda}$.

In total, we have branched over $k^{\bigoh(k)}$ extracts for $t$, $k^{\bigoh(k^2)}$ linear orders of $Z$, and $k^{\bigoh(k)}$ choices of $\lambda$, resulting in a total of $k^{\bigoh(k^2)}$ instances. What remains now is to show how the solution of each instance $I_{\alpha,j,\lambda}$ can be used to solve \textsc{IMB Join}. Let $i_{\alpha,j,\lambda}$ denote the output of $I_{\alpha,j,\lambda}$. Then we set $c(\alpha)$ to the minimum value of $i_{\alpha,j,\lambda}+\imbv_\lambda$ over all choices of $\alpha, j, \lambda$.

In conclusion, we argue correctness. For a contradiction, assume that there exist some $\alpha,j,\lambda$ such that $i_{\alpha,j,\lambda}+\imbv_\lambda<c(\alpha)$. By our construction, there then exists a linear order $R_1$ of $X_t\cup \bigcup_{t'\in A_t} v\in Y_{t'}$ such that the imbalance of all vertices in $\bigcup_{t'\in A_t}$ (induced by edges of $G[Y_t]$ with at least one endpoint in $\bigcup_{t'\in A_t} v\in Y_{t'}$) is equal to $\imbv_\lambda$, and furthermore this $R_1$ satisfies the ``requirements'' of the extract $\alpha$ on the imbalance of vertices in $\partial(Y_t)\setminus X_t$. By the construction of the instance $I_{\alpha,j,\lambda}$, there also exists a linear order $R_2$ of $X_t\cup \bigcup_{t'\in B_t} v\in Y_{t'}$ where the order of vertices in $X_t$ is the same as in $R_1$, and hence it is possible to merge $R_1$ and $R_2$ into a linear order $R$ over $Y_t$ by using vertices in $X_t$ as ``anchoring points'' ($R$ preserves the order of the anchoring points and the order of vertices within $R_1$ and of vertices within $R_2$, and the order between a non-anchoring vertex in $R_1$ and a non-anchoring vertex in $R_2$ is irrelevant). One can straightforwardly verify that the linear constraints $\zeta$ in $I_{\alpha,j,\lambda}$ ensure that $R$ realizes the extract $\alpha$. The cost of $R$ is then equal to the $\imbv_\lambda$ plus the imbalance of vertices in $X_t$ and the imbalance of vertices in $\bigcup_{t'\in B_t}Y_{t'}$. The mapping $\omega$ transfers the information on the imbalance of vertices in $X_t$ caused by edges in $G[X_t\cup \bigcup_{t'\in A_t}Y_{t'}]$ into $I_{\alpha,j,\lambda}$, and from there on one can verify that the imbalance of $R$ in vertices in $X_t\cup \bigcup_{t'\in B_t}Y_{t'}$ sums up to $i_{\alpha,j,\lambda}$. Hence this linear order $R$ contradicts the value of $c(\alpha)$.

On the other hand, assume $i_{\alpha,j,\lambda}+\imbv_\lambda>c(\alpha)$ for all $j,\lambda$. Let $R$ be a linear order which realizes $\alpha$; by reversing the merging procedure outlined in the previous paragraph, one can decompose $R$ into $R_1$ (over $X_t \bigcup_{t'\in A_t} v\in Y_{t'}$) and $R_2$ (over $X_t\cup \bigcup_{t'\in B_t} v\in Y_{t'}$). Let $j$ be the unique linear suborder of $R_1$ over $Z$, and $\lambda$ the unique tuple of extracts capturing $R_1$ on individual children $t'\in A_t$. Then, by our assumption on the size of $c(\alpha)$, either the imbalance of $R_1$ over $\bigcup_{t'\in A_t} v\in Y_{t'}$ is greater than $\imbv_\lambda$, or the imbalance of $R_2$ over $X_t\cup \bigcup_{t'\in B_t} v\in Y_{t'}$ (additionally counting edges between $X_t$ and $\bigcup_{t'\in A_t}$) is greater than $i_{\alpha,j,\lambda}$. However, the first can be ruled out by the computation of $\imbv_\lambda$ from $\lambda, j, \alpha$, while the second is impossible since one can plug in $R_2$ into $I_{\alpha,j,\lambda}$ to achieve a solution value of $i_{\alpha,j,\lambda}$.

We conclude that $i_{\alpha,j,\lambda}+\imbv_\lambda=c(\alpha)$, and hence by iterating over all values of $\alpha$ and given the correct value of $i_{\alpha,j,\lambda}$, it is possible to construct $\cD(t)$ in time $k^{\bigoh(k^2)}\cdot |V(G)|$.
\end{proof}

\begin{lemma}
\label{lem:redimbFPT}
There exists an algorithm which solves \textsc{Reduced IMB Join} in time $k^{\bigoh(k^4)}\cdot (|B_t|+1)$.
\end{lemma}

\begin{proof}
We again use an ILP formulation, but this time we first need to get rid of the absolute values over $\rhd_R(v)-\lhd_R(v)+\omega(v)$. Since the number of vertices $v$ in $X_t$ is bounded by $k$, we can exhaustively branch over whether each $\rhd_R(v)-\lhd_R(v)+\omega(v)$ ends up being negative or non-negative, and force this choice in our ILP instance by suitable constraints. Formally, let us branch over all the at most $2^k$ possible functions $\phi:v\in X_t\rightarrow \{1,-1\}$. For each fixed $\phi$, we construct an ILP instance $I_\phi$ which outputs the minimum value of $\big(\sum_{v\in Y_t\setminus X_t}\imbv_R(v)\big)+\big(\sum_{v\in X_t} \phi(v)\cdot( \rhd_R(v)-\lhd_R(v)+\omega(v))\big)$ (over all $R$ which satisfy the required conditions); let this value be $i_\phi$. Then the solution of \textsc{Reduced IMB Join} can be correctly computed as $\min_{\phi:v\in X_t\rightarrow \{1,-1\}}i_{\phi}$.

Before proceeding, we mention a few important observations. First, each vertex $v\in Y_{t'\in B_t}$ can be placed in one of the at most $k+1$ intervals in the linear order between the placement of vertices in $X_t$. Second, since there are no edges between any $Y_{t'\in B_t}$ and $Y_{t''\in B_t}$, the order among vertices in different children of $t$ cannot change the value of $i_{\phi}$; only the order between the neighbors of $X_t$ and $X_t$ itself matters. Third, for each $t'\in B_t$, we know that the minimum sum of imbalances in $Y_{t'}$ over all linear orders is $a_{t'}$, and by Lemma~\ref{lem:imbdiscard} it suffices to only consider extracts $\gamma$ of $Y_{t'}$ such that $\beta_t(\gamma)\neq \infty$. 

We now describe the construction of an ILP instance which outputs $i_{\phi}$, as well as its relation to a linear order $R$ so as to facilitate the correctness argument. Let $\cL'_0,\cL'_1,\cL'_2$ denote the set of all extracts of nodes $t'$ where $|\partial(Y_{t'})|=0,1,2$, respectively. Observe that while each of the sets $\cL'_1$, $\cL'_2$ may contain a number of extracts that cannot be bounded by any function of our parameter (since each node formally has its own, unique extract), we may identify and define equivalences that group the extracts of nodes which can affect their neighborhood in exactly the same way. In particular, analogously as in the proof of Lemma~\ref{lem:redcvcfpt} we say that extract $(f_a,\tau_a)\in \cL'_i$ of a node $a$ is \emph{equivalent} to extract $(f_b,\tau_b)\in \cL'_i$ of a node $b$, $i\in \{1,2\}$, if there exists a bijective mapping $\phi$   
from $\partial(Y_a)$ to $\partial(Y_b)$ such that (1) $N(v)\setminus Y_a = N(\phi(v))\setminus Y_b$ for every $v\in \partial(Y_a)$,  
(2) $u<_{f_a} v$ if and only if $\phi(u) <_{f_b} \phi(v)$ for every $u,v\in \partial(Y_a)$, and 
(3) $\tau_a(v)=\tau_b(\phi(v))$ for every $v\in \partial(Y_a)$. In other words, two extracts are equivalent if there is a canonical renaming function 
between vertices in the boundary of $Y_a$ and $Y_b$  which preserves the neighborhoods outside and transforms the extract $(f_a,\tau_a)$ into $(f_b,\tau_b)$. 
We can now set $\cL_0=\cL'_0$, while $\cL_1\subseteq \cL'_1$ and $\cL_2\subseteq \cL'_2$ are obtained by keeping precisely one representative extract for each equivalence class that appears in $\cL_1$ and $\cL_2$, respectively (in other words, all but one arbitrarily selected extract in each equivalence class is deleted). It can be observed that the size of each such $\cL_j$ is upper-bounded by $\bigoh(k^2)$.

Next, we let the set $S_0$ contain all the constantly many mappings $\beta_0:\cL_0\rightarrow [0]\cup\{\infty\}$. Similarly, we let $S_{1,1}$ contain all $\beta_1:\cL_1\rightarrow [4]\cup\{\infty\}$ (extracts in $S_1$ can appear in nodes of adhesion $1$ or $2$), and we let $S_{1,2}$ contain all $\beta_1:\cL_1\rightarrow [8]\cup\{\infty\}$. For the last case, we let $S_2$ contain all $\beta_2:\cL_2\rightarrow [8]\cup\{\infty\}$. We let:
\begin{itemize}
\item $\&_{1,1}=S_{1,1}\times X_t$, 
\item $\&_{1,2}=S_{1,2}\times X_t\times X_t$, 
\item $\&_{2,1}=S_2\times X_t$, 
\item $\&_{2,2}=S_2\times X_t\times X_t$.
\end{itemize}
 The set $\&=S_0\cup \&_{1,1}\cup \&_{1,2}\cup \&_{2,2}\cup \&_{2,1}$ will serve as the set of all possible ``types'' of nodes $t'\in B_t$, while the set $S=S_0\cup S_{1,1}\cup S_{1,2}\cup S_2$ will contain all possible mappings $\beta$. For each $\beta\in S_{1,2}\cup S_2$, nodes $t'$ have two neighbors in $X_t$; these will in general not be symmetric with respect to $\beta$, and hence we distinguish them by denoting one as $N_1(t')$ and the other as $N_2(t')$. Before proceeding, note that $|\&|\in \bigoh(k^2)$.

For each $\pi\in \&$, we introduce variables which capture information on the distribution of the (at most two) vertices in $\partial(Y_{t'})$ among the $k+1$ intervals of $f$. We formalize for the most general case of $2$ vertices in $\partial(Y_{t'})$ and $2$ neighbors in $X_t$, i.e., for $\&_{2,2}$; the remaining three cases are simplifications of this one and require less variables.
For each $\pi=(\beta,x_1,x_2)\in \&_{2,2}$, we define $Q_{\pi}=\SB t'\in B_t\SM \exists a_{t'}: \cD(t')=(a_{t'},\beta), N_1(t')=x_1, N_2(t')=x_2 \SE$. We let $q_\pi=|Q_{\pi}|$, i.e., $q_\pi$ is the number of nodes of type $\pi$. For a node $t'$, let $N(x_1)\cap \partial(Y_{t'})$ be denoted $y_1$ and analogously $N(x_2)\cap \partial(Y_{t'})=y_2$. Now we introduce (for this $\pi$) $(k+1)\cdot(k+2)$ new variables: 
\begin{itemize}
\item for each $0\leq i, j\leq k$, if $i\neq j$ we introduce a variable $g_i^{h_j^\pi}$, which captures the number of nodes $t'$ of type $\pi$ such that $y_1$ occurs after exactly $i$ vertices of $X_t$ in $R$ and $y_2$ occurs after exactly $j$ vertices of $X_t$ in $R$, and
\item if $i=j$ we introduce two new variables $g_i^{<h^\pi_j}$ and $g_i^{>h^\pi_j}$, which capture the number of pairs of vertices such that both $y_1$ and $y_2$ occur after exactly $i$ vertices of $X_t$ and $y_1$ occurs, respectively, before or after $y_2$.
\end{itemize}

In total, we have introduced $\bigoh(k^4)$ variables. Next, we insert the following constraints:
\begin{enumerate}
\item For each $\pi\in \&$, we make sure that the variables indeed correspond to a partition of the set $Q_{\pi}$. Specifically, we make sure that each variable is greater or equal to $0$, and that the sum of all variables associated with $\pi$ is equal to $q_\pi$.
\item We add a constraint for each linear constraint in $\zeta$ on the value of $\rhd(v)-\lhd(v)$, where $v\in X_t$. This value of $\rhd(v)-\lhd(v)$ can be expressed by a linear combination over our variables, since each variable is associated with a fixed position (left or right) from each vertex in $X_t$ and a type, which in turn contains information on whether the variable is adjacent or not to each $v\in X_t$.
\item We make sure that we only consider parameter values consistent with $\phi$. Specifically, if $\phi(v)=1$ (non-negative) then we need to ensure that $\rhd_R(v)-\lhd_R(v)+\omega(v)\geq 0$, while in the other case we need to ensure that $\rhd_R(v)-\lhd_R(v)+\omega(v)< 0$. Each of these conditions can be expressed as a linear constraint over our variables, similarly as when encoding $\zeta$.
\end{enumerate}

The constructed ILP instance can be solved in time at most $k^{\bigoh(k^4)}\cdot (|B_t|+1)$ by Theorem~\ref{thm:pilp}.

Let $a$ denote $\sum_{t'\in A_t}a_{t'}$. For each variable $var$ corresponding to a placement of a type $\pi$ in the at most $2$ intervals, one can compute from $\beta\in \pi$ and the order between $y_1,y_2,x_1,x_2$ a value $\cost(var)$ in $[8]\cup \infty$ which expresses the additional imbalance in $Y_{t'}$ caused by this particular placement. Then our instance will minimize the expression $val(\phi)=a+\sum_{v\in X_t} \phi\cdot (\rhd_R(v)-\lhd_R(v)+\omega(v))+ \sum_{var}\cost(var)$. After branching through all possible values of $\phi$, we output the minimum over all computed $val(\phi)$.
\end{proof}

\begin{corollary}
\label{cor:imbjoin}
There exists an algorithm which solves \textsc{IMB Join} in time $k^{\bigoh(k^4)}\cdot n$.
\end{corollary}

Now the proof of the theorem below is then analogous to the proof of Theorem~\ref{thm:CVCalg}.

\begin{theorem}
\label{thm:IMBalg}
\imb~can be solved in time $k^{\bigoh(k^4)}\cdot n^2$.
\end{theorem}

\begin{proof}
We use Theorem~\ref{thm:decompsize} to transform $(T,\mc{X})$ into a nice tree-cut decomposition with at most $2n$ nodes. We then use Lemma~\ref{lem:imbinit} to compute $\cD(t)$ for each leaf $t$ of $T$, and proceed by computing $\cD(t)$ for nodes in a bottom-to-top fashion by Corollary~\ref{cor:imbjoin}. The total running time per node is dominated by $n\cdot k^{\bigoh(k^4)}$ and there are at most $2n$ nodes. Once we obtain $\cD(r)$, we can correctly output by Observation~\ref{obs:imbend}.
\end{proof}

\subsection{Capacitated Dominating Set}
\label{sub:cds}
\textsc{Capacitated Dominating Set} is a generalization of the classical \textsc{Dominating Set} problem by the addition of vertex capacities. Aside from its applications (discussed for instance in~\cite{KuhnMoscibroda10}), \textsc{Capacitated Dominating Set} has been targeted by parameterized complexity research in the past also because it is a useful tool for obtaining W-hardness results~\cite{FominGolovachLokshtanovSaurabh09,FialaGolovachKratochvil11,BodlaenderLokshtanovPenninkx09}.
It is known to be W[1]-hard when parameterized by treewidth~\cite{DomLokstanovSaurabhVillanger08}.

Let $G=(V,E)$ be a capacitated graph with a capacity function $c:V(G) \rightarrow \Nat_0$. We say that $D\subseteq V(G)$ is a {\em capacitated dominating set} of $G$ if there exists a mapping $f:V\setminus D \rightarrow D$ which maps every vertex to one of its neighbors so that the total number of vertices mapped by $f$ to any $v\in D$ does not exceed $c(v)$. 
Such mapping $f$ is said to {\em witness} $D$ and $f$ is a witness function of $D$. We formally define the problem below.

\begin{center}
  \begin{boxedminipage}[t]{0.99\textwidth}  
\begin{quote}
  tcw-\textsc{Capacitated Dominating Set} (\cds)\\ \nopagebreak
  \emph{Instance}: A capacitated graph $G$ on $n$ vertices together with a width-$k$ tree-cut decomposition $(T,\mc{X})$ of $G$, and an integer $d$.  \\ \nopagebreak
  \emph{Parameter}: $k$. \\ \nopagebreak
  \emph{Task}: Decide whether there exists a capacitated dominating set $D$ of $G$ containing at most $d$ vertices. \nopagebreak
\end{quote}
\end{boxedminipage}
\end{center}

For a vertex $v\in \partial(Y_t)$, recall that $\adh_t(v)=|N(v)\setminus Y_t|$ and that $\adh(t)=\sum_{v\in\partial(Y_t)}\adh_t(v)$.


We now give a high-level description of our algorithm. In a ``snapshot'' of a capacitated dominating set at a node $t$ of the decomposition, a vertex in $\partial(Y_t)$ may have two possible states: {\em active} or {\em passive}. The interpretation is as follows: a vertex $v\in \partial(Y_t)$ is
\begin{itemize}
\item active if it will either be in the dominating set or will be dominated by a vertex in $Y_t$, and
\item passive if it will be dominated by a vertex outside of $Y_t$.
\end{itemize}


We will use a data table $\cD(t)$ where each entry stores information about the states of vertices in $\partial(Y_t)$, i.e., the vertices which have neighbors outside of $Y_t$.  Moreover, the snapshot also contains information about the residual capacity of each vertex in $\partial(Y_t)$ that is \emph{active}; in particular, each such vertex $v$ in $\partial(Y_t)$ may have up to $\adh_t(v)\leq k$ neighbors outside of $Y_t$, and hence it is important to distinguish whether the residual capacity of $v$ is at least $\adh_t(v)$ (meaning that $v$ can dominate all of its neighbors outside of $Y_t$), or precisely some number $\ell$ in $[\adh_t(v)-1]$ (meaning that $v$ can only dominate up to $\ell$ of its neighbors outside of $Y_t$)\footnote{Notice that active vertices that are not in the dominating set automatically receive a residual capacity of $0$, but are otherwise not distinguished from vertices in the dominating set.}. We store the information about these residual capacities in the form of an ``offset''. For each such snapshot, we will keep information about the minimum-size capacitated dominating set that corresponds to that snapshot. We formalize below.

\newcommand{\act}{\emph{active}}
\newcommand{\off}{\emph{offset}}

\subsubsection{Data Table}
In this subsection we define the table $\cD(t)$ at node $t$, but before that we will formalize the notion of a \emph{snapshot} (for $t$):

\begin{definition}
A snapshot is a tuple $(\emph{active},\emph{offset})$ where $\emph{active}\subseteq \partial(Y_t)$ and $\emph{offset}$ maps each $v\in \emph{active}$ to an element of $[\adh_t(v)]$.
\end{definition}

For a node $t\in V(t)$, let the \emph{representation} of a snapshot $\sigma=(\emph{active},\emph{offset})$ be the graph $G_\sigma$ obtained from $G[Y_t]$ by 1) deleting all vertices in $\partial(Y_t)\setminus \emph{active}$ and 2) attaching to each vertex $v\in \emph{active}$ precisely $\emph{offset}(v)$ many new pendant vertices; we will assume that each new pendant vertex created in this way has a capacity of $0$ and refer to these new pendant vertices as \emph{auxiliary vertices}. The  \emph{cost} of the snapshot $\sigma$, denoted $\emph{cost}(\sigma)$, is then the minimum size of a capacitated dominated set  over all capacitated dominating sets of $G_\sigma$; any capacitated dominating set of $G_\sigma$ contains no auxiliary vertex. We call such a minimum capacitated dominating set a $\sigma$-CDS, and if no such capacitated dominating set exists then we set $\emph{cost}(\sigma)=\bot$. The \emph{base cost} of $t$, denoted $a_t$, is then defined as the cost of the snapshot $(\emptyset,\emptyset)$.

The final ingredient we need before defining our data table $\cD(t)$ is a bound on the gap between the cost of an arbitrary snapshot and the base cost.

\begin{lemma}
\label{lem:cdsbounds}
For each snapshot $\sigma=(\emph{active},\emph{offset})$ it holds that either $\emph{cost}(\sigma)=\bot$, or $a_t \leq \emph{cost}(\sigma)\leq a_t+2\adh(t)$.
\end{lemma}

\begin{proof}
For the first inequality, it suffices to observe that every $\sigma$-CDS is also an $(\emptyset,\emptyset)$-CDS. 
We prove the second inequality by induction on $|\act|+\sum_{v\in \act}\off(v)$. 
The proof is a two-step induction, where the first induction step covers the addition of a vertex and the second induction step covers an increase of the offset. 
For the base case $(\emptyset,\emptyset)$, 
we observe that the claim clearly holds for $|\emph{active}|=0$ by the definition of $a_t$.

For the first inductive step, assume that the second inequality holds for some $\sigma=(\act,\off)$, and let us consider $\sigma'=(\act\cup\{v\},\off\cup\{v\mapsto 0\})$. Then either there exists a $\sigma'$-CDS of the same size as a $\sigma$-CDS (notably, when there exists a $\sigma$-CDS containing $v$), or for an arbitrary $\sigma$-CDS $X$ we have that $X\cup \{v\}$ is a $\sigma'$-CDS. In other words, adding a vertex to $\act$ (with an initial $\off$ of $0$) only increases the cost of the snapshot by at most $1$.

For the second inductive step, assume that the second inequality holds for some $\sigma=(\act,\off)$, and let us consider $\sigma'=(\act,\off')$ where $\off'(v)=\off(v)+1$ for one vertex $v\in \act$ and $\off'(w)=\off(w)$ for every other vertex $w\in \act$. We now distinguish two cases. If $|N_{G_{\sigma'}}(v)\setminus Y_t|>c(v)$ then clearly no $\sigma'$-CDS can exist, i.e., $\emph{cost}(\sigma')=\bot$. Otherwise we are left with three subcases that are similar to the first induction step: 
\begin{itemize}
\item either a $\sigma$-CDS is also a $\sigma'$-CDS, or 
\item we can take an arbitrary $\sigma$-CDS $X$ not containing $v$ and observe that $X\cup \{v\}$ is a $\sigma'$-CDS, or 
\item we can take an arbitrary $\sigma$-CDS $X$ containing $v$ and observe that since $X$ is not a $\sigma'$-CDS, $v$ must be used to dominate some vertex, say $z$. We then observe that $X\cup \{z\}$ is a $\sigma'$-CDS.
\end{itemize}

To summarize, we have shown that for every snapshot $\sigma$, increasing the $\off$ or $\act$ will only increase the size of a $\sigma$-CDS by $1$; since both $|\act|$ and the sum of offsets is upper-bounded by $\adh(t)$, the lemma follows.
\end{proof}

We can now define our data table:
\begin{definition}
\label{def:CDStable}
$\cD(t)=(a_t,\beta_t)$ where $\beta_t$ maps each snapshot of $t$ to a value in $[2\adh(t)]\cup\{\bot\}$ such that for each snapshot $\sigma$ it holds that $\emph{cost}(\sigma)=a_t+\beta_t(\sigma)$ (where $\bot$ acts as an absorbing element).
\end{definition}

The correctness of Definition~\ref{def:CDStable} follows from Lemma~\ref{lem:cdsbounds}. 

\subsubsection{Initialization and Termination}
\begin{lemma}
\label{lem:cdsinit}
Let $t$ be a leaf in a nice tree-cut decomposition $(T,\mc{X})$ of a graph $G$, and let $k$ be the width of $(T,\mc{X})$. Then $\cD(t)$ can be computed in time $2^{\bigoh(k^2)}$.
\end{lemma}

\begin{proof}
We can branch over all at most $2^k\cdot 2^k$-many snapshots of $t$, and for each snapshot $\sigma$ a $\sigma$-CDS can be computed by brute force in time at most $2^{\bigoh(k^2)}$.
\end{proof}

For the next observation, it will be useful to note that the root can always be assumed to contain an empty bag (otherwise, one may attach a new root with this property on top of the original one).

\begin{observation}
\label{obs:cdsend}
Let $(G,d)$ be an instance of \cds~and let $r$ be the root of a nice tree-cut decomposition of $G$ such that $X_r=\emptyset$ and $\cD(r)=(a_r,\beta_r)$. Then $(G,d)$ is a yes-instance if and only if $a_r\leq d$.
\end{observation}

\subsubsection{Inductive Step}
What remains is to show how to compute $\cD(t)$ for a node $t$ once $\cD(t')$ is known for each child $t'$ of $t$. We formalize this problem below.

\begin{center}
  \begin{boxedminipage}[t]{0.99\textwidth}
\begin{quote}
  \textsc{CDS Join}\\ \nopagebreak
  \emph{Instance}: A \cds~instance $(G,d)$, a non-leaf node $t$ of a width-$k$ nice tree-cut decomposition $(T,\mc{X})$ of $G$, and $\cD(t')$ for each child $t'$ of $t$.\\
    \emph{Parameter:} $k$. \\ \nopagebreak
  \emph{Task}: Compute $\cD(t)$. \nopagebreak
\end{quote}
\end{boxedminipage}
\end{center}

For a third time, we will use the two-step approach of first reducing to a
``simpler'' problem and then applying a suitable ILP encoding. 
We call the problem we reduce to \textsc{Reduced CDS Join}.

\begin{center}
  \vspace{0.3cm}
   \begin{boxedminipage}[t]{0.99\textwidth}
\begin{quote}
  \textsc{Reduced CDS Join}\\ \nopagebreak
  \emph{Instance}: A \cds~instance $(G,d)$, a non-leaf root $t$ of a width-$k$ nice tree-cut decomposition $(T,\mc{X})$ of $G$ such that $A_t=\emptyset$ and $G[X_t]$ is edgeless, $\cD(t')$ for each child $t'$ of $t$, and auxiliary sets $S\subseteq X_t$, $S'\subseteq V(G)$ with the property that $\forall s'\in S' \exists q\in B_t: Y_q=\{s'\}$.\\
    \emph{Parameter:} $k$. \\ \nopagebreak
  \emph{Task}: Determine whether $G$ admits a capacitated dominating set $D$ of size at most $d$ such that $D\cap X_t=S$ and $D\cap S'=\emptyset$. \nopagebreak
  \end{quote}
\end{boxedminipage}
\end{center}

\begin{lemma}
\label{lem:cdsreduce}
There is an FPT Turing reduction from \textsc{CDS Join} to $2^{\bigoh(k^2)}$ instances of \textsc{Reduced CDS Join} which runs in time $2^{\bigoh(k^2)}\cdot n\cdot d$. 
\end{lemma}

\begin{proof}
In order to compute $\cD(t)$ for a \textsc{CDS Join} instance, it suffices to be able to compute the $\emph{cost}$ of each snapshot of $t$. Since the number of snapshots of $t$ is upper-bounded by $2^{\bigoh(k)}$, we can loop over each snapshot in this time and reduce the problem of computing the cost of each individual snapshot to \textsc{Reduced CDS Join}. So let us consider an arbitrary snapshot $\sigma=(\act,\off)$ of $t$.

As our first step, we remove all vertices outside of $Y_t$ (thus turning $t$ into a root). We then branch over each subset $S\subseteq X_t$ with the aim of identifying the intersection between a $\sigma$-CDS and $X_t$; to this end, we only consider subsets $S$ which are ``compatible'' with $\act$, notably by requiring that $S\cap \partial(Y_t)\subseteq \act$ (i.e., vertices in $S$ that lie on the boundary cannot be passive). Moreover, each vertex $x$ in $X_t\setminus (\partial(Y_t)\cup S)$ must be dominated by a neighbor in $Y_t$, and we also branch to determine whether $x$ will be dominated by a specific vertex in $S$ or rather by a vertex in $Y_{t'}$ for some unspecified child $t'$ of $t$. If $x$ is dominated by some $x'\in S$, then we delete $x$ and reduce the capacity of $x'$ by $1$. Afterwards, we delete all edges with both endpoints in $X_t$, since we have at this point predetermined the precise dominating relations within $X_t$. For a fixed $\sigma$, these branching steps require us to consider at most $2^{\bigoh(k^2)}$ many subcases. Observe that in the current branch, every vertex in $X_t\setminus (\partial(Y_t)\cup S)$ now must be dominated by a neighbor in $Y_t$ and that every vertex in $\partial(Y_t)\setminus \act$ must be left undominated.

Our next and crucial step is aimed at dealing with children in $A_t$. In particular, for each $p\in A_t$, we branch over each snapshot $\sigma_p=(\act',\off')$ of $p$ and check whether $\sigma_p$ is compatible with $S$ and $\sigma$: 
\begin{itemize}
\item for each $v\in \partial(Y_t)\cap \partial(Y_p)$ it must hold that if $v\not \in \act$ then $v\not \in \act'$\footnote{On the other hand, it may happen that $v\in \act\setminus \act'$, since $v$ might be dominated by a vertex in $Y_t\setminus Y_p$.}, and
\item if $v\in \act\cap \act'$ then $\off(v)\leq \off'(v)$.
\end{itemize}
As a subcase, we now deal with vertices $v\in \partial(Y_t)\cap \partial(Y_p)$ such that $v\in \act'\setminus \act$. The interpretation here is that $v$ cannot be a dominating vertex (since $v\not \in \act'$), but must be dominated from a neighbor in $Y_t\setminus Y_p$. Since there are at most $k$ such neighbors of $v$, we brute-force to determine which neighbor dominates $v$ and reduce that neighbor's capacity by $1$.

Next, for each $v\in \partial(Y_p)\setminus \partial(Y_t)$, we distinguish the following two cases. If $v\not \in \act'$, then $v$ must be dominated by a neighbor in $Y_t\setminus Y_p$; there are at most $k$ such neighbors and we brute-force to determine which one dominates $v$, whereas for each choice we reduce the dominating vertex's capacity by $1$. If $v\in \act'$ and $\off'(v)>0$ then $v$ may be used to dominate some of its (at most $k$) neighbors in $Y_t\setminus Y_p$, and we brute-force to determine which vertices it dominates and delete the newly dominated vertices.
Moreover, for each $v\in \partial(Y_p)\cap \partial(Y_t)$ such that $\off'(v)-\off(v)=\ell$, $v$ can be used to dominate up to $\ell$ of its neighbors in $Y_t\setminus Y_p$; we once again brute-force to determine which and delete these newly dominated vertices	.

Finally, we need to ensure that a solution for our constructed instance of \textsc{Reduced CDS Join} satisfies the requirements specified by the $\off$ component of $\sigma$. To this end, for each $v\in \act$ we construct $\off(v)$-many new pendant vertices, attach these to $v$, and add these to the set $S'$ (ensuring that these $\off(v)$-many vertices must be dominated by $v$). Since these vertices must also be included in the provided nice tree-cut decomposition of $G$, we simply add each such pendant to its own separate (thin) leaf adjacent to $t$.

Let us now consider the set $\cY$ of instances of \textsc{Reduced CDS Join} obtained by branching over one particular snapshot $\sigma$ of $t$ in the original input instance of \textsc{CDS Join}. Assume that the cost of $\sigma$ is $\delta$, as witnessed by a $\sigma$-CDS $Z$ with witness function $f$. Let us now consider the branch where $S=Z\cap X_t$ and the pairs of vertices in a domination relationship inside $X_t$ reflect the witness function $f$. Similarly, for each $p\in A_t$, consider the branch where we identified the snapshot $\sigma_p$ that mimics the behavior of $Z$---in particular by having $\act'$ reflect $Z\cap Y_t$ and $\off'$ reflect $f$. Then, by the correctness of $\cD(p)$ and optimality of $Z$, it follows that $|Z\cap Y_p|=a_p+\beta_p(\sigma_p)$. Moreover, $|Z\cap X_t|=|S|$. The intersection $Z\cap (\cup_{q\in B_t}Y_q)$ then represents a capacitated dominating set for an instance $Y\in \cY$ (obtained in our reduction) of size $\delta-|S|-(\sum_{p\in A_t}a_p+\beta_p(\sigma_p))$.

On the other hand, let us consider that for some snapshot $\sigma$ leading to a set $\cY$ of instances of \textsc{Reduced CDS Join} via branching, there is an instance $Y\in \cY$ which admits a capacitated dominating set $Z'$ satisfying the stated properties for $S$ and $S'$ of cardinality $d'$. After adjusting $Z'$ by taking into account the selection of $S$ and the domination relations on $X_t$ that led to $Y$, and similarly by expanding $Z'$ via the $\sigma_p$-CDS' used in the branching for the individual nodes $p\in A_t$, we can reverse the arguments used in the previous paragraph and show that there is a $\sigma$-CDS in the original instance of \textsc{CDS Join} of size at most $|Z'|+|S|+(\sum_{p\in A_t}a_p+\beta_p(\sigma_p))$.

To conclude, we have shown that the original instance of \textsc{CDS Join} can be solved by computing a minimum capacitated dominating set for each of the at most $2^{\bigoh(k^2)}$ obtained \textsc{Reduced CDS Join} instances, whereas the construction steps for each obtained instance can be carried out in linear time. Such a minimum set can be computed by performing at most $d$ separate calls asking for a capacitated dominating set with the required properties of size at most $d$.
\end{proof}

\begin{lemma}
\label{lem:redcdsfpt}
\textsc{Reduced CDS Join} can be solved in time $2^{\bigoh(k\cdot \log k)}|B_t|$.
\end{lemma}

\begin{proof}
Observe that to solve an instance of \textsc{Reduced CDS Join}, it suffices to identify, for each child $q$ of $t$, a snapshot of a capacitated dominating set for $G[Y_q]$ while taking into account that (1) vertices in $X_t\setminus S$ must be dominated by a neighbor in $Y_q$ for some child $q$ of $t$, and (2) a vertex $x$ in $S$ can be used to dominate $c(x)$-many vertices in the sets $Y_q$. We will resolve point (1) by branching and then use an ILP formulation for point (2).

As our first step, let us define an equivalence $\equiv$ which identifies which children of $t$ ``behave the same way'' as far as the potential interaction between their dominating sets and $X_t$. Formally, given two children $p,q$ of $t$ such that $\cD(p)=(a_p,\beta_p)$ and $\cD(q)=(a_q,\beta_q)$, $p\equiv q$ if and only if there exists a bijective renaming function $\iota:\partial(Y_p)\rightarrow \partial(Y_q)$ such that:
\begin{itemize}
\item for each $v\in \partial(Y_p)$, $N(v)\setminus Y_p=N(\iota(v))\setminus Y_q$ (i.e., neighborhoods in $X_t$ are preserved), and
\item for each snapshot $\sigma_p$ of $p$, $\beta_p(\sigma_p)=\beta_q(\sigma_q)$ where $\sigma_q$ is obtained by applying $\iota$ component-wise on $\sigma_p$ (i.e., offsets of snapshots are preserved).
\end{itemize}

Since the total number of possible snapshots (up to renaming) for thin nodes is upper-bounded by a constant, the number of functions $\beta_q$ for each child $q\in B_t$ is also upper-bounded by a constant. Hence there are at most $\bigoh(k^2)$-many equivalence classes of $\equiv$; let us denote these $F_1,\dots,F_{y}$ where the number of equivalence classes $y$ is in 
$\bigoh(k^2)$. 

We will now perform exhaustive branching to identify, for each $x\in X_t\setminus S$ (i.e., a vertex in $X_t$ which must be dominated by a child of $t$), an equivalence class of $\equiv$ containing a node $p$ such that a vertex in $\partial(Y_p)$ will dominate $x$, along with the used snapshot of $p$. Since $|X_t|\leq k$ and both the number of equivalence classes as well as the number of snapshots are bounded as above, this amounts to a branching factor of at most $\bigoh(k^{2k})=2^{\bigoh(k\cdot \log k)}$. In each branch, we simply check that the selected snapshot of $p$ can indeed dominate $x$ (by checking the value of the $\off$ in the selected snapshot $\sigma_p$); if not then we discard the current branch, and otherwise we remove $p$ and reduce $d$ by $a_p+\beta_p(\sigma_p)$, where $\cD(p)=(a_p,\beta_p)$. During the branching, we take into account the fact that such $p$ can be used to dominate up to $2$ different vertices in $X\setminus S$---in particular, from the second vertex in $X\setminus S$ onward, we allow for a previously deleted child of $t$ to be picked once again. 

The above branching ensures that all vertices in $X_t$ are dominated at this point, and it suffices to define the constraints for an ILP that will identify the number of times a snapshot should be used in each equivalence class of $\equiv$ in order to obtain a capacitated dominating set of size at most $d$. Let $\Omega=\sum_{p\in B_t;\cD(p)=(a_p,\beta_p)}a_p$. For the variables of the ILP, we proceed as follows:

First, for each equivalence class $F_i$ containing $\#_i$-many children of $t$, let $\sigma^1,\dots,\sigma^z$ be the snapshots that occur on these children (modulo the bijective renaming function as defined above). For each $i\in [y]$ and $j\in [z]$, the ILP will contain an integer variable $s_i^j$ which captures the number of nodes in $F_i$ for which the sought-after capacitated dominating set $D$ will correspond to the $j$-th snapshot, i.e., for each node $q\in F_i$ the set $D\cap Y_q$ will have the same intersection on $\partial(Y_q)$ as specified by $\act$ and the capacities on $\partial(Y_q)$ will correspond to those given in $\off$. 

Second, let $S=\{x_1,\dots,x_\iota\}$ for $\iota\leq k$. For each each variable $s_i^j$ which corresponds to a snapshot of an equivalence class whose borders all contain a single vertex, say $w$, that is adjacent to two distinct vertices $x_a,x_b\in S$ and which must be dominated from $X_t$ (i.e., $w\not \in \act$), we create variables $x_{a,i}^j$ and $x_{b,i}^j$ which will capture how many of the vertices in the border of the nodes in $F_i$ will be dominated by $x_a$ and by $x_b$, respectively. These are the only variables that appear in the ILP, and hence (unlike in the previous two problems) it only has constantly-many variables. Let $\chi_a$ be a shorthand for $\sum_{x_{a,i}^j\text{ exists and is defined as above}}x_{a,i}^j$.

Before introducing the constraints, for each $x_\ell\in S$ let $Q_\ell$ be the set of all variables $s_i^j$ corresponding to the snapshots of the equivalence classes which contain precisely \emph{one} vertex on the boundary that (1) lies outside of $\act$ and (2) has precisely one neighbor in $X_t$, and that is $x_\ell$; in other words, $Q_\ell$ contains those variables which correspond to children in $B_t$ that will require $x_\ell$ to dominate one vertex. Similarly, let $U_\ell$ be the set of all variables $s_i^j$ corresponding to the snapshots of the equivalence classes which contain precisely \emph{two} vertices on the boundary that (1) both lie outside of $\act$ and (2) both have precisely one neighbor in $X_t$, and that is $x_\ell$; in other words, $U_\ell$ contains those variables which correspond to children in $B_t$ that will require $x_\ell$ to dominate two vertices.

We now introduce the following constraints in the ILP:
\begin{enumerate}
\item Each variable is non-negative;
\item For each equivalence class $F_i$, we ensure that variables $s_i^1,\dots,s_i^{z}$ represent a correct partitioning of the nodes in $F_i$, in particular by requiring $\#_i=\sum_{j\in [z]}s_i^j$;
\item For each vertex $x_\ell\in X_t$, we ensure that $x_\ell$ has sufficient capacity to dominate all the vertices in the borders of children in $B_t$ given the domination requirements of the individual snapshots and the number of times each snapshot occurs in each $F_i$, in particular by requiring $c(x)\geq \big(\sum_{r\in Q_\ell}r\big)+\big(\sum_{u\in U_\ell}2u\big)+\chi_\ell$;
\item Finally, for each variable $s_i^j$ which corresponds to a snapshot of an equivalence class whose borders all contain a single vertex, say $w$, that is adjacent to two distinct vertices $x_a,x_b\in S$, we ensure that the nodes which will reject snapshot $\sigma_j$ in equivalence class $F_i$ are all dominated by requiring $s_i^j=x_{a,i}^j+x_{b,i}^j$.
\end{enumerate}

Correctness follows via arguments that are analogous to those in the proof of Lemma~\ref{lem:redcvcfpt}. Since the ILP formulation has constant size, it can be solved by Theorem~\ref{thm:pilp} in time at most $|B_t|$ and the lemma follows.
\end{proof}

We now have all the ingredients necessary to establish the fixed-parameter tractability of \cds.

\begin{theorem}
\label{thm:CDSPalg}
\cds~can be solved in time $2^{\bigoh(k^2)}\cdot n^2$.
\end{theorem}

\begin{proof}
The proof is analogous to the proof of Theorem~\ref{thm:CVCalg} and Theorem~\ref{thm:IMBalg}.
\end{proof}

\subsection{Proof of Proposition~\ref{prop:twtcw}}

As the final result in this section, we provide a proof for the relationship between treewidth and tree-cut width claimed in Subsection~\ref{sub:Rel}.

To provide some intuition for the proof, let us first recall that for each thin node $t$ of a nice tree-decomposition with parent $p(t)$, $t$ may belong to $B_{p(t)}$ or $A_{p(t)}$ 
depending on whether $N(Y_t)\subseteq X_{p(t)}$ holds or not. 
Let us say that a thin node $t$ is \emph{$B$-thin} (resp. \emph{$A$-thin}) if $t\in B_{p(t)}$ (resp. $t\in A_{p(t)}$). 

In the proof, we show that a nice tree-cut decomposition $(T,\mc{X})$ of width at most $k$ can be converted to a tree-decomposition of width of $2k^2+3k$. 
One intuitive way to obtain a tree-decomposition from $(T,\mc{X})$ would be to, for every edge $xy$ of $G$, add $\{x,y\}$ to all nodes of $T$ 
lying on the unique path connecting $t_x$ and $t_y$, namely the nodes whose bags contain $x$ and $y.$ 
Put equivalently, for each vertex $x$ of $G$, one can find a minimal subtree $T(x)$ of $T$ whose bags collectively cover $N[x]$ 
and add $x$ to all nodes of $T(x)$. It is easy to check that the resulting collection of bags together with $T$ form a tree-decomposition. 
However, this construction can create a bag with an arbitrarily large number of vertices. Specifically, 
any node $t$ of a nice tree-cut decomposition $(T,\mc{X})$ may have unboundedly many $B$-thin children, and the aforementioned ``intuitive'' construction would add to $t$ all vertices in $\bigcup_{t'\in B_t} Y_t$ with neighbors in $X_t.$ 

To remedy this issue, we use a {\sl truncated} version of $T(x)$ and add $x$ to the nodes of this truncated version of $T(x)$ only. 
In this way, we intend to avoid adding $x$ to $t$ when $x\in Y_{t'}$ for some $B$-thin child $t'$ of $t$ even if $t$ belongs to $T(x).$ 
It turns out that this simple tweak of the previous attempt is sufficient to obtain the desired bound on the treewidth. 
We formalize the idea in the proof below.

\begin{proof}[Proof of Proposition~\ref{prop:twtcw}]
Let $G$ be an arbitrary graph of tree-cut width $k$. By Theorem~\ref{thm:decompsize}, we may assume that $G$ has a nice tree-cut decomposition $(T,\mc{X})$ with at most $2|V(G)|$ nodes.

From $(T,\mc{X})$, we construct a pair $(T, \mc{Z})$, where $\mc{Z}$ consists of vertex subsets $Z_t \subseteq V(G)$ over all nodes $t$ of $T$. 
For every vertex $v\in V(G)$, let $t(v)$ be the node of $T$ whose bag $X_{t(v)}$ contains $v$ and 
let $T(v)$ be the minimal subtree of $T$ satisfying $N[v]\subseteq \bigcup_{t\in T(v)}X_t.$
Clearly, $t(v)$ is contained in $T(v).$ If there is a $B$-thin node on the unique path from the root of $T(v)$ to $t(v)$, 
let ${\sf top}^*(v)$ be the lowermost $B$-thin node among such $B$-thin nodes; 
if no $B$-thin node exists between $t(v)$ and ${\sf top}(v)$, we set ${\sf top}^*(v)$ to be the root of $T(v).$ Now, 
for every $v\in V(G)$ let $T'_v$ be the subtree of $T(v)$ rooted at ${\sf top}^*(v).$ 
Finally, for each node $t$ of $T'_v$, we add $v$ to $Z_t$. This completes the construction of $(T, \mc{Z}).$ 


We want to show that $(T, \mc{Z})$ is a tree-decomposition of $G$. First, 
the construction guarantees that for every $v\in V(G)$, the set of bags containing $v$ is the tree $T'_v$, thus is connected in $T$. 
Next, we shall argue that for every edge $uv$ of $G,$ there exists a bag in $\mc{Z}$ containing both $u$ and $v.$ For this, 
it suffices to verify that $T'_u$ and $T'_v$ have at least one node in common; let us assume for a contradiction that this is not the case.
If the root of $T'_v$ is an ancestor of $T'_u$ in $T$, then 
$T'_v$ must have included $t(u)$ due to the edge $uv$ and thus the root of $T'_u$. 
Therefore, if $T'_u$ and $T'_v$ are disjoint, the roots of $T'_u$ and $T'_v$ cannot be in an ancestor-descendant relation.   
Note that both $T(u)$ and $T(v)$ contain $\{t(u),t(v)\}$, and in particular the least common ancestor of $T'_u$ and $T'_v$. 
The fact that $T'_u\neq T(u)$ and $T'_v\neq T(v)$ entails that the roots, i.e., ${\sf top}^*(u)$ and ${\sf top}^*(v)$ of $T'_u$ and $T'_v$, respectively, are $B$-thin nodes. 
Recall that $t(u)$ and $t(v)$ are respectively nodes of $T'_u$ and $T'_v$, hence $u\in Y_{{\sf top}^*(u)}$ and $v\in Y_{{\sf top}^*(v)}$. 
However, the niceness property of $(T,\mc{X})$ implies that $N(Y_{{\sf top}^*(u)})$ are included in the bag of the parent of ${\sf top}^*(u),$ a contradiction. 
Therefore, it follows that $(T, \mc{Z})$ is a tree-decomposition of $G$. 

To verify the width of $(T, \mc{Z}),$ consider an arbitrary node $t$ and let us bound the number of vertices $v$ such that $t\in T'_v,$ 
which equals $|Z_t|$. 
Observe that $t\in T'_v$ holds only if $t$ is a descendant of ${\sf top}^*(v)$ (possibly the same node). 
This condition does not hold for any vertex in $Y_{t'}$ for $t'\in B_t,$ implying that $Z_t$ is disjoint from $\bigcup_{t'\in B_t} Y_t.$ 
Consider $v\in Z_t \setminus X_t$. If $v\notin Y_t,$ observe that $t(v)$ is a strict ancestor of $t$ and 
the reason that $v$ is placed in $Z_t$ is because $T'_v$ contains $t$, and especially $v$ has a neighbor in $Y_t.$ 
As such an edge connecting $v$ and its neighbor in $Y_t$ is counted in $\adh(t)$, there are at most $|\adh(t)|$ such vertices $v,$ namely 
satisfying $v\in Z_t\setminus X_t \setminus Y_t.$
If $v\in Y_t,$ note that $v\in Y_{t'}$ for some $t'\in A_t$ and $v\in Z_t$ implicates that $T'_v$ contains $t$. 
This means that $v$ has a neighbor outside $Y_{t'}$ since otherwise, the root of $T'_v$ would have been a descendant of $t'.$
Such an edge connecting $v$ and its neighbor outside $Y'_{t'}$ is counted in $\adh(t')$. With~\cref{lem:Asmall}, there are at most $2k+1$ nodes in $A_t$ 
and hence, there are at most $(2k+1)\cdot k$ vertices $v$ such that $v\in (Z_t \setminus X_t)\cap Y_t.$ This yields the claimed bound 
$|Z_t|\leq |X_t| + |\adh(t)| + \sum_{t'\in A_t}|\adh(t')| \leq 2k^2+3k.$  
\end{proof}

\section{Lower Bounds}
\label{sec:hard}


We show that \textsc{List Coloring}~\cite{ErdosRubinTaylor79} and \textsc{Precoloring Extension}~\cite{BiroHujterTuza92} are $\W[1]$-hard
parameterized by tree-cut width, strengthening the known
$\W[1]$-hardness results with respect to 
treewidth~\cite{FellowsEtAl11}.
Both problems have been studied extensively in the classical \cite{Marx05,GravierKoblerKubiak02} as well as parameterized \cite{FellowsEtAl11,Ganian11} settings. A \emph{coloring} $c$ is a mapping from the vertex set of a graph to a set of colors; a coloring is \emph{proper} if for every pair of adjacent vertices $a,b$, it holds that $c(a)\neq c(b)$. Problem definitions follow.

\begin{center}
  \begin{boxedminipage}[t]{0.99\textwidth}
\begin{quote}
  tcw-\textsc{List Coloring}\\ \nopagebreak
  \emph{Instance}: A graph $G=(V,E)$, a width-$k$ tree-cut decomposition $(T,\mc{X})$ of $G$, and for each vertex $v\in V$ a list $L(v)$ of permitted colors. \\ \nopagebreak
  \emph{Parameter:} $k$. \\ \nopagebreak
  \emph{Task}: Decide whether there exists a proper vertex coloring $c$ such that $c(v)\in L(v)$ for each $v\in V$. \nopagebreak
\end{quote}
\end{boxedminipage}
\end{center}

The tcw-\textsc{Precoloring Extension} problem may be defined
analogously as tcw-\textsc{List Coloring}; the only difference is that in
\textsc{Precoloring Extension} lists are restricted to either contain
a single color or all possible colors. Before stating the new hardness result, we mention (and prove) the following observation.

\begin{observation}
\label{obs:coldeg}
\textsc{List Coloring} and \textsc{Precoloring Extension} parameterized by $\mathbf{degtw}$ are FPT.
\end{observation}

\begin{proof}
  Assume that there exists a solution to an instance of \textsc{List
    Coloring}, i.e., there exists a mapping $col$ from each vertex to
  a color in $L(v)$; w.l.o.g., we assume that colors are represented
  by numbers in $\mathbb{N}_0$. From $col$ we can construct a new
  mapping $col'$ such that $col'$ maps each vertex $v$ to one of the
  first $\mathbf{degtw}+1$ elements in $L(v)$ without creating
  conflicts (this may be done greedily from $col$ in an iterative
  fashion). The mapping $col'$ witnesses that there must exist a
  solution to the instance which only chooses one of the first $\mathbf{degtw}+1$
  colors in each list. This allows the construction of a standard
  dynamic algorithm on the tree-decomposition of the input graph.
\end{proof}

\begin{theorem}
\label{thm:listhard}
tcw-\textsc{List Coloring} and tcw-\textsc{Precoloring Extension} are $\W[1]$-hard.
\end{theorem}

\begin{proof}  
  We use the reduction from the $\W[1]$-hard problem \textsc{Multi-Colored Clique} (\textsc{MCC}) to \textsc{List Coloring}
  described in \cite[Theorem~2]{FellowsEtAl11}. 
\begin{center}
 \begin{boxedminipage}[t]{0.99\textwidth}
\begin{quote}
 \textsc{Multi-Colored Clique (MCC)}\\ \nopagebreak
 \emph{Instance}: A $k$-partite graph $G$ with $k$ parts $V_1,\ldots , V_k$, each consisting of $n$ vertices.  \\ \nopagebreak
 \emph{Parameter:} $k$. \\ \nopagebreak
 \emph{Task}: Decide whether there a $k$-clique in $G$. \nopagebreak
\end{quote}
\end{boxedminipage}
\end{center}
We outline the
  reduction below. Given an instance $G=(V_1\cup\dots\cup V_k, E)$ of
  \textsc{MCC}, we construct an instance $I$ of \textsc{List Coloring}
  with vertex sets $X,Y$ whereas for each $V_i$ there is a single
  $i\in X$ such that $L(i)=V_i$. Then for each non-edge $\{a\in
  V_i,b\in V_j\}\not \in E$, where $i\neq j$, we construct a vertex
  $y\in Y$ such that $y$ is adjacent to $i$ and $j$ and
  $L(y)=\{a,b\}$. Then the choice of color for each $i\in X$
  corresponds to a choice of a single vertex from $V_i$, and the set
  $Y$ contains ``constraints'' which prevent the choice of two
  non-adjacent vertices.

To prove that \textsc{List Coloring} is $\W[1]$-hard, it now suffices to prove that the constructed instance $I$ has tree-cut width at most $k$. To this end, consider the following tree-cut decomposition $(T,\mc{X})$ of $I$: $T$ is a star with center $s\in V(T)$ and $X_s=X$, and for each $y\in Y$ there is a leaf $z\in V(T)$ such that $X_z=\{y\}$. Assume that $T$ is rooted at $s$. Then $\adh(s)=0$ and $\adh(z)=2$ for each leaf $z$. Hence $\tor(s)=k$ and $\tor(z)=3$, from which it follows that $I$ has tree-cut width $k$.

To show that \textsc{Precoloring Extension} is also $\W[1]$-hard, we recall the simple reduction from \textsc{List Coloring} to \textsc{Precoloring Extension} also described in \cite{FellowsEtAl11}. From an instance of \textsc{List Coloring}, it is possible to construct an instance of \textsc{Precoloring Extension} by ``modeling'' the lists through the addition of precolored pendant vertices (vertices of degree one). Let $I'$ be an instance of \textsc{Precoloring Extension} constructed in this manner from an instance $I$ of \textsc{List Coloring} described above; let $Q$ contain all the new pendant vertices, i.e., $Q=V(I')\setminus V(i)$. 

We show that $I'$ also has a tree-cut decomposition $(T',\mc{X'})$ of width $k$. For vertex sets $X,Y\in I'$ we use $(T,\mc{X})$. Then for each $q\in Q$ adjacent to a vertex $v\in X\cup Y$ such that $v\in X_{t_v}$, we construct a pendant $t_q\in V(T')$ adjacent to $t_v$ such that $X'_{t_q}=\{q\}$. The adhesion and torso-size of vertices in $X\cup Y$ remains the same as in $(T,\mc{X})$, while for each $q\in Q$ it holds that $\adh(t_q)=1$ and $\tor(t_q)=1$. The theorem follows.
\end{proof}

We also show that the \textsc{Constraint Satisfaction Problem} (\textsc{CSP}) is $\wx{1}$-hard when parameterized by the tree-cut width of the incidence graph, even when restricted to the Boolean domain; this is not the case for $\mathbf{degtw}$~\cite{SamerSzeider10a}. Formal definitions follow.

An instance $I$ of the \textsc{CSP} is a triple~$(X, D, \CCC)$, where $X$ is a finite set of \emph{variables}, $D$ is finite set of \emph{domain values}, and $\CCC$ is a finite set of \emph{constraints}.  Each constraint in~$\CCC$ is a pair $(S,R)$, where~$S$, the \emph{constraint scope}, is a non-empty sequence of distinct variables of~$X$, and $R$, the \emph{constraint relation}, is a relation over~$D$ (given as a set of tuples) whose arity matches the length of~$S$.  A CSP instance $(X, D, \CCC)$ is \emph{Boolean} if $D=\{0,1\}$.  An \emph{assignment} is a mapping from the set $X$ of variables to the domain~$D$. An assignment $\tau$ \emph{satisfies} a constraint $C=((x_1,\dots,x_n),R)$ if $(\tau(x_1),\dots,\tau(x_n))\in R$, and $\tau$ satisfies the CSP instance if it satisfies all its constraints. An instance~$I$ is  \emph{satisfiable} if it is satisfied by some assignment.

The \emph{incidence graph} $G_I$ of CSP instance $I=(V, D, \CCC)$ is the bipartite graph whose vertex set is formed by $V\cup \CCC$, and where a constraint $C=(S,R)\in \CCC$ is incident exactly to all the variables in~$S$.

\begin{center}
   \begin{boxedminipage}[t]{0.99\textwidth}
\begin{quote}
   tcw-\textsc{CSP}\\ \nopagebreak \emph{Instance}: A CSP instance
  $I=(X,D,\CCC)$  together with a width-$k$
  tree-cut decomposition $(T,\mc{X})$ of the incidence graph $G_I$ of~$I$.\\ \nopagebreak \emph{Parameter:} $k$. \\ \nopagebreak
  \emph{Task}: Decide whether $I$ is satisfiable.
 \nopagebreak
\end{quote}
\end{boxedminipage}
\end{center}

tcw-\textsc{Boolean-CSP} denotes tcw-\textsc{CSP} restricted to Boolean CSP instances.  We note that \textsc{Boolean-CSP} parameterized by $\mathbf{degtw}$ is fixed-parameter tractable \cite{SamerSzeider10a}. 

\begin{theorem}
\label{thm:csphard}
tcw-\textsc{Boolean CSP} is  $\W[1]$-hard.
\end{theorem}

\begin{proof}
  We give a reduction from \textsc{MCC}.  Let $(G,k)$ be an instance
  of \textsc{MCC} with $V(G)=\bigcup_{i=1}^n V_i$ where
  $V_i=\{v_{i,1},\dots,v_{i,n}\}$ for $1\leq i \leq k$. We first
  construct a CSP instance $I=(X,D,\CCC)$ and will show later how to
  make it Boolean.  We let $X= \SB x_{i,j} \SM 1\leq i,j \leq k\SE$
  and $D=\{1,\dots,n\}$.  For $1\leq i \leq k$ the set $\CCC$ contains
  the constraint
  $C_i^=((x_{i,1},\dots,x_{i,k}),\{(0,\dots,0),(1,\dots,1),\dots,(n,\dots,n)\})$, which
  enforces that the values of all the variables
  $x_{i,1},\dots,x_{i,k}$ are the same.  Furthermore, for each $1\leq
  i<j\leq k$ the set $\CCC$ contains the constraint
  $C^E_{i,j}=((x_{i,j},x_{j,i}),\SB (a,b)\in [n] \times [n] \SM
  v_{i,a}v_{j,b} \in E(G)\SE$ which encodes the incidence relation
  between $V_i$ and~$V_j$. It follows by construction that $I$ is
  satisfiable if and only if $G$ contains a $k$-clique.
  
  Next we obtain from $I$ a Boolean CSP instance
  $I'=(X',\{0,1\},\CCC')$ by replacing each variable $x_{i,j}$ with $n$
  Boolean variables $x_{i,j}^{(1)},\dots,x_{i,j}^{(n)}$. Intuitively,
  assigning $x_{i,j}$ the value $a$ corresponds to assigning
  $x_{i,j}^{(a)}$ the value $1$ and all other $x_{i,j}^{(b)}$ for
  $b\neq a$ the value $0$ (instead of this unary encoding we could
  have used a more succinct binary encoding, but the unary encoding is
  simpler and suffices for our purposes). Consequently, each constraint
  $C_i^=$ of $I$ (of arity $n$) gives rise to a constraint $C_i^{='}$
  of $I$ (or arity $n^2$), and each constraint $C^E_{i,j}$ of $I$ (or
  arity 2) gives rise to a constraint $C^{E'}_{i,j}$ of $I'$ (of arity
  $2n$).  By construction, $I$ is satisfiable if and only if $I'$ is
  satisfiable.

  In order to complete the reduction and the proof of the theorem, it
  remains to construct a tree-cut decomposition of the incidence graph
  $G_{I'}$ of $I'$ of a width that is bounded by a function of $k$.
  We observe that $I'$ has exactly $k'=k+\binom{k}{2}$ many
  constraints and each variable appears in the scopes of exactly two
  constraints. Hence one side of the bipartite graph $G_{I'}$ consists
  of $k'$ many vertices and the other side only of vertices of
  degree~2. This already implies that the tree-cut with of $G_{I'}$ is
  at most $k'$, as we can take the following tree-cut decomposition
  $(T,\mc{X})$ of $G_{I'}$.  $T$ is a star with center $s\in V(T)$ where
  $X_s=\CCC$, and for each variable $x_{i,j}^{(a)}$ there is leaf
  $t_{i,j}^{(a)}\in V(T)$ with
  $X_{t_{i,j}^{(a)}}=\{x_{i,j}^{(a)}\}$.
\end{proof}

 
\section{Concluding Notes}
We have provided the first algorithmic applications of the new graph
parameter tree-cut width, considering various of hard combinatorial
problems. In some cases we could establish fixed-parameter
tractability, in some cases we could establish $\W[1]$-hardness,
staking off the potentials and limits of this parameter (see
Table~\ref{tab:overview}). 

The FPT algorithms make use of our new
notion of nice tree-cut decompositions, which we believe to be of
independent interest. In fact, following their introduction these decompositions have already been used to obtain algorithms for problems such as \textsc{Bounded Degree Vertex Deletion}~\cite{GanianKO21}, \textsc{Stable Roommates with Ties and Incomplete Lists}~\cite{BredereckHKN19}, and \textsc{Edge Disjoint Paths}~\cite{GanianO21}. Surprisingly, while the third problem is \XP\ when parameterized by tree-cut width, it remains $\W[1]$-hard under this parameterization~\cite{GanianO21}---a stark contrast to the behavior of \textsc{Vertex Disjoint Paths} parameterized by treewidth.

We remark that the quadratic dependency on $n$ in our algorithms can almost certainly be reduced to linear by using a carefully designed data structure or avoiding the use of Turing reductions. In particular, in all three cases the quadratic dependency is caused by the fact that the Turing reductions used in the proofs of Lemmas~\ref{lem:cvcreduce}, \ref{lem:imbreduce} and \ref{lem:cdsreduce} produce subinstances of size $\bigoh(n)$; other than that, these reductions run in time that depends only on $k$.

While we do not yet have an exact fixed-parameter algorithm for computing tree-cut width, the 2-approximation algorithm of Kim et al.~\cite{KimOumPaulSauThilikos14} is sufficient to establish fixed-parameter tractability of various problems of interest. Moreover, there is also a SAT encoding which can compute the tree-cut width exactly for graphs with dozens of vertices~\cite{GanianLOS19}.

\section*{Acknowledgments}
This work was partially supported by the Austrian Science Fund (FWF),
projects P31336 (NFPC) and Y1329 (ParAI), and the Vienna Science and Technology Fund
(WWTF), project ICT19-065 (REVEAL). The authors also gratefully acknowledge the detailed feedback provided by the reviewers.

\bibliographystyle{abbrv} \bibliography{literature}

\end{document}